\documentclass[11pt]{article}
\usepackage{fullpage}
\usepackage{amssymb,amsmath}
\usepackage{epsfig}
\usepackage{palatino}
\usepackage{graphicx}
\usepackage{textcomp}
\usepackage{amsthm}
\usepackage{mdframed}
\usepackage{color}
\usepackage{tocloft}
\usepackage[colorlinks=true,linkcolor=blue,linktoc=all]{hyperref}
\newtheorem{theorem}{Theorem}
\newtheorem{proposition}[theorem]{Proposition}

\newtheorem{definition}[theorem]{Definition}
\newtheorem{algorithm}[theorem]{Algorithm}
\newtheorem{claim}[theorem]{Claim}
\newtheorem{lemma}[theorem]{Lemma}

\newtheorem{corollary}[theorem]{Corollary}
\newtheorem{fact}[theorem]{Fact}

\newtheorem*{theorem*}{Theorem}
\newtheorem*{lemma*}{Lemma}

\newcommand{\nc}{\newcommand}
\nc{\rnc}{\renewcommand}

\def\ba#1\ea{\begin{align}#1\end{align}}
\def\bas#1\eas{\begin{align*}#1\end{align*}}
\def\bpm#1\epm{\begin{pmatrix}#1\end{pmatrix}}
\nc{\nn}{\nonumber}
\nc{\eq}[1]{(\ref{eq:#1})}
\nc{\eqs}[2]{(\ref{eq:#1}) and (\ref{eq:#2})}

\def\begsub#1#2\endsub{\begin{subequations}\label{eq:#1}\begin{align}#2\end{align}\end{subequations}}
\nc\qand{\qquad\text{and}\qquad}
\nc\mnb[1]{\medskip\noindent{\bf #1}}

\nc\benum{\begin{enumerate}}
\nc\eenum{\end{enumerate}}
\nc\bit{\begin{itemize}}
\nc\eit{\end{itemize}}
\nc{\ot}{\otimes}
\rnc{\L}{\left} 
\nc{\R}{\right}

\def\Z{{\mathbb{Z}}}
\def\R{\mathbb{R}}
\def\N{\mathbb{N}}
\def\C{\mathbb{C}}

\def\mod{\mbox{mod}}
\def\vol{\mbox{vol}}

\def\poly{{\rm poly}}
\def\log{{\rm log}}

\def\ASD{\rm ASD}
\def\NC{\rm NC}
\def\RNC{\rm RNC}

\newcommand{\be}{\begin{eqnarray}}
\newcommand{\ee}{\end{eqnarray}}

\newcommand\floor[1]{{\lfloor #1 \rfloor}}
\newcommand\ceil[1]{{\lceil #1 \rceil}}
\newcommand\round[1]{{\lfloor #1 \rceil}}

\def\P{{\sf{P}}}

\newcommand{\ignore}[1]{}

\newcommand{\eps}{\varepsilon}
\renewcommand{\epsilon}{\varepsilon}

\nc{\hin}{h_{\text{in}}}
\nc{\pin}{\partial_{\text{in}}}
\nc{\pell}{\partial_{\ell}}

\newcommand{\nocontentsline}[3]{}
\newcommand{\tocless}[2]{\bgroup\let\addcontentsline=\nocontentsline#1{#2}\egroup}
\makeatletter
\newcommand{\cftsectionprecistoc}[1]{\addtocontents{toc}{%
  {\leftskip \cftsecindent\relax
   \advance\leftskip \cftsecnumwidth\relax
   \rightskip \@tocrmarg\relax
   \textit{#1}\protect\par}}}
\makeatother

\begin{document}

\title{\bf A Quasi-Random Approach to Matrix Spectral Analysis}

\author{Michael Ben-Or\thanks{School of Computer Science and
Engineering, The Hebrew University,
Jerusalem, Israel}
 \and Lior Eldar\thanks{Center for Theoretical Physics, MIT, Email: leldar@mit.edu}}

\date{\today}
\maketitle

\abstract{Inspired by the quantum computing algorithms for Linear Algebra problems \cite{HHL,TaShma}
we study how the simulation on a classical computer of this type of ``Phase Estimation algorithms" performs
when we apply it to solve the Eigen-Problem of Hermitian matrices. The result is a completely new, efficient
and stable, parallel algorithm to compute an approximate spectral decomposition of any Hermitian matrix.
The algorithm can be implemented by Boolean circuits in $O(\log^2 n)$ parallel time with
a total cost of $O(n^{\omega+1})$ Boolean operations. This Boolean complexity matches the best known
rigorous $O(\log^2 n)$ parallel time algorithms, but unlike those algorithms our algorithm is (logarithmically) stable,
so further improvements may lead to practical implementations.

All previous efficient and rigorous approaches to solve the Eigen-Problem use randomization
to avoid bad condition as we do too. Our algorithm makes further use of randomization in a completely new way,
taking random powers of a unitary matrix to randomize the phases of its eigenvalues. 
Proving that a tiny Gaussian perturbation and a random polynomial power are sufficient
to ensure almost pairwise independence of the phases $(\mod\  2\pi)$ is the main technical contribution of this work.
This randomization enables us, given a Hermitian matrix with well separated eigenvalues,
to sample a random eigenvalue and produce an approximate eigenvector
in $O(\log^2 n)$ parallel time and $O(n^\omega)$ Boolean complexity.
We conjecture that further improvements of our method can provide a stable solution to
the full approximate spectral decomposition problem with complexity similar
to the complexity (up to a logarithmic factor) of sampling a single eigenvector.
%
}

\tocless
\section{Introduction}

\subsection{General}

The eigen-problem of Hermitian matrices is the problem of computing the eigenvalues and eigenvectors of a Hermitian matrix.
This problem is ubiquitous in computer science and engineering, and because of its relatively high computational complexity
imposes a high computational load on most modern information processing systems.

The Abel-Ruffini theorem implies there is no deterministic expression for the eigenvalues of a
general matrix $A$, and for this reason eigen-problem algorithms must be iterative.
This gives rise to a host of problems: these algorithms are hard to analyze rigorously,
and often turn out to be unstable, and hard to parallelize.

As with many other problems in computer science, one typically considers an approximate spectral decomposition of a matrix.
Thus, given a matrix $A$, we are usually interested not in its exact eigenvalues / eigenvectors, which may be very hard to compute,
(and possibly very long to describe once computed), but rather in an approximate decomposition:
\begin{definition}\label{def:approx}

\textbf{Approximate Spectral Decomposition} - $\ASD(A,\delta)$

\noindent
Let $A$ be some $n\times n$ Hermitian matrix.
An approximate spectral decomposition of $A$, with accuracy parameter $\delta =1/\poly(n)$ is a set of unit vectors $\{v_i\}_{i=1}^n, \|v_i\| = 1$ such that
each $v_i$ has $\|A v_i - \lambda_i v_i \|_2 \leq \delta$, for some number $\lambda_i\in \R$,
and
$$
\left\| \sum_{i\in [n]} \lambda_i v_i v_i^T - A \right\| \leq \delta \left\|A\right\|,
$$
where $\|X\|$ is the operator norm of $X$.
\end{definition}

By standard arguments the above can be generalized to an arbitrary $n \times n$ matrix $A$, by considering the Hermitian matrix $A^H A$, in which case
$\ASD(A^H A,\delta)$ is an approximation of the {\it singular vectors} (and singular values) of $A$.
We note that the definition of $\ASD$ then corresponds to a "smooth analysis" of matrices: namely given input $A$, we do not
find a spectral decomposition of $A$, but rather the decomposition of a matrix $A'$, such that $\| A - A'\| \leq \delta$.
We also point out, that the definition of $\ASD$ holds just as well in the case of nearly degenerate matrices: 
we do not require a one-to-one correspondence with the eigenvectors of $A$,
which can be extremely hard to achieve, but rather to find some set of approximate eigenvectors, such that the corresponding
weighted sum of rank-$1$ projections form an approximation of $A$.

When one considers an algorithm ${\cal A}$ for the $\ASD$ problem, one can examine its {\it arithmetic} complexity
or {\it boolean} complexity.
The arithmetic complexity is the minimal size arithmetic circuit $C$ (namely each node computes addition/multiplication/division to unbounded accuracy)
that implements ${\cal A}$, whereas the boolean complexity counts the number of boolean AND/OR gates of fan-in $2$ required to implement ${\cal A}$.

Given the definition above, and following Demmel et al. \cite{DDH07} we consider an algorithm ${\cal A}$ to be log-stable (or stable for short),
if there exists a  circuit $C$ that implements ${\cal A}$ on $n\times n$ matrices, and a number $t = O(\log(n))$, such that each arithmetic computation in 
$C$ uses at most $t$ bits of precision, and output of the circuit deviates from the output of the arithmetic circuit by at most $1/\poly(n)$.
We note that when an algorithm is {\it stable} then its boolean complexity is equal to its arithmetic complexity up to a factor $O(\log(n))$.
If, however, an algorithm is {\it unstable} then its boolean complexity could be larger by a factor of up to $n$.
In the study of practical numerical linear algebra algorithms, one usually identifies algorithms that are 
stable with "practical", and algorithms that are not stable to be impractical.
This usually, because the computing machines are restricted to representing numbers with a number of bits that is a small
fraction of the size of the input.

In terms of parallelism, we will refer to the complexity class $\NC^{(k)}$ (see Definition \ref{def:nc}) which is the set of all computational problems
that can be solved 
by uniform Boolean circuits of size $\poly(n)$ in time $O(\log^k(n))$.
Often, we will refer to the class $\RNC^{(k)}$, in which the parallel $\NC^{(k)}$ circuit is also allowed to accept uniform random bits.
One would like an ASD algorithm to have minimal arithmetic / boolean complexity, and minimal parallel time.
Ideally, one would also like this algorithm to be stable.

\subsection{Main Contribution}
Inspired by recent quantum computing algorithms \cite{HHL,TaShma}, we introduce a new perspective on the problem of computing the $\ASD$ that is based on low-discrepancy sequences.
Roughly speaking, low-discrepancy sequences are deterministic sequences which appear to be random, because they "visit" each small sub-cube the same number of times that a completely random sequence would
up to a small additive error.
\begin{definition}

\textbf{Multi-dimensional Discrepancy}

\noindent
For integer $s$, put $I^s = [0,1)^s$.
Given a sequence $x = (x_n)_{n=1}^{N}$, with $x_n \in I^s$ the discrepancy $D_M(x)$ is defined as:
$$
D_M(x) = \sup_{B\in {\cal B}} \left\{ \left| \frac{1}{M}\sum_{n=1}^M \chi_B(x_n) - \vol(B) \right| \right\},
$$
where $\chi_B(x_n)$ is an indicator function which is $1$ if $x_n\in B$ is the set of all $s$-products of intervals $\prod_{i=1}^s [u_i,v_i]$, with $[u_i,v_i] (\mod 1)\subseteq [0,1)$.
\end{definition}

We recast the $\ASD$ problem as a question about
the discrepancy of a certain sequence related to the input matrix.
Specifically, given a Hermitian matrix $A$ with, say $n$ unique eigenvalues $\{ \lambda_i\}_{i\in [n]}$ 
the central object of interest is the sequence comprised of 
$n$-dimensional vectors of eigenvalue residuals:
$$
S(A) = 
\left(\{\lambda_1 \cdot 1\}, \hdots, \{ \lambda_n \cdot 1 \}\right),
\left(\{\lambda_1 \cdot 2\}, \hdots, \{ \lambda_n \cdot 2 \}\right),
 \hdots
\left(\{\lambda_1 \cdot M\}, \hdots, \{ \lambda_n \cdot M \}\right),
$$
where $\{x\}$ is the fractional part of $x\in \R$, and $M = \poly(n)$ is some large integer.
$S(A)$ is hence a sequence of length $M$ in $[0,1)^n$.
We would like $S(A)$ to have as small discrepancy as possible.
Hence, in sharp contrast to previous algorithms, instead of the computational effort being concentrated on revealing "structure" in the matrix, our algorithm is actually focused
on producing random-behaving dynamics.

The main application of our approach presented in this paper is a new stable and parallel and stable algorithm for computing the $\ASD$ of any Hermitian matrix.
\begin{mdframed}
\begin{theorem}
For any Hermitian matrix $0 \preceq A\preceq 0.9 I$, and $\delta \leq n^{-7}$ we have $\ASD(A,\delta)\in RNC^{(2)}$,
with total boolean complexity $O(n^{\omega+1})$.
The algorithm is $\log$-stable.
\end{theorem}
\end{mdframed}

\noindent
\\
The boolean complexity of our algorithm is $O(n^{\omega+1})$.
If however, one is interested in sampling a uniformly random eigenvector, it can be achieved in complexity $O(n^{\omega})$.
\footnote{$\omega$ signifies the infimum over all constants $c$ such that one can multiply two matrices in time at most $n^c$.}

\subsection{Overview of the Algorithm}

To compute the $\ASD$ of a given matrix $A$, we first consider a similar problem of sampling uniformly an approximate eigenvector of $A$,
where the eigenvalues of $A$ are assumed to be well-separated.
Clearly, if one can sample from this distribution in $\RNC^2$, then by the coupon collector's bound concatenating
$O(n\log(n))$ many parallel copies of this routine, one can sample all eigenvectors quickly with high probability.
To do this, we require a definition of a Hermitian matrix that is $\delta$-separated:
\begin{definition}

\textbf{$\delta$-separated}

\noindent
Let $A$ be an $n\times n$ PSD matrix.
$A$ has a complete set of real eigenvalues $\lambda_1 > \lambda_2 > \hdots > \lambda_n\geq 0$.
We say that $A$ is $\delta$-separated if $\lambda_j - \lambda_{j+1} \geq \delta$ for all $j<n$, and $\lambda_1 \leq 1 - \delta$.
\end{definition}

\noindent
Next, we introduce the notion of a separating integer w.r.t. a sequence of real numbers:
\begin{definition}\label{def:sep}

\textbf{Separating Integer}

\noindent
Let $\bar{\lambda} = (\lambda_1,\hdots, \lambda_n)\in [0,1)^n$.
For $\alpha >  4$ define
$$
B_{out} = [1 - 1/(4n), -1 + 1/(4n)]
\mbox{\quad and \quad}
B_{in}(\alpha) = [-1/(\alpha n ), 1/(\alpha n)],
$$
A positive integer $m$ is said to separate the $k$-th element of $\bar{\lambda}$ w.r.t. $B_{in},B_{out}$ if it satisfies:
\begin{itemize}
\item
$\{ m \lambda_k\} \in B_{in}(\alpha)$
\item
$
\forall j\neq k \ \ \{ m \lambda_j\} \notin B_{out}$
\end{itemize}
\end{definition}

\noindent
and finally define the notion of a separating integer w.r.t. a $\delta$-separated matrix.
\begin{definition}\label{def:sep1}
A positive integer $m$ is said to separate  $k$ in a $\delta$-separated matrix $A$ w.r.t. $B_{in},B_{out}$,
if $m$ separates the $k$-th element of ${\cal L}(A)$ w.r.t. $B_{in},B_{out}$.
\end{definition}

Following is a sketch of the main sampling routine.  For complete details see Section \ref{sec:filter}.
The routine accepts a separating integer $m$  of the $i$-th eigenvalue of a $\delta$-separated matrix $A$,
a precision parameter $\delta$ and returns a $\delta$ approximation
of the $i$-th eigenvector of $A$:
\begin{mdframed}

\begin{algorithm}\textbf{${\rm Filter}(A,m,\delta)$}

\noindent
\begin{enumerate}
\item
\textbf{Compute parameters:}
$$
p = 24 n^2 \ceil{\ln (1/\delta)},
\zeta =  \delta^2 / (2 p m).
$$
\item
\textbf{Sample random unit vector:}

\noindent
Sample a standard complex Gaussian vector $v$, set $w_0 = v / \|v\|$.
\item
\textbf{Approximate matrix exponent: }\label{it:comp1}

\noindent
Compute a $\zeta$ Taylor approximation of $e^{i A}$, denoted by $\tilde U$.
\item
\textbf{Raise to power:}

\noindent
Compute $\tilde U^m$ by repeated squaring.
\item
\textbf{Generate matrix polynomial:}

\noindent
Compute $B = \left(\frac{I + \tilde U^m}{2}\right)^p$ by repeated squaring.
\item
\textbf{Filter:}\label{it:comp4}

\noindent
Compute 
$
w
=
\frac{B \cdot w_0}{\| B \cdot w_0\|}.
$
\item
\textbf{Decide:}\label{decide}

\noindent
Set $z = A \cdot w$, $i_0 = \arg\max_{i\in [n]} |w_i|$ and compute
$
c = z_{i_0} / w_{i_0}.
$
If
$$
\left\| A \cdot w - c \cdot w \right\| \leq 3\delta \sqrt{n}
$$
return $w$, and
otherwise reject.
\end{enumerate}

\end{algorithm}
\end{mdframed}

\noindent
\\
In words - the algorithm samples a random vector and then multiplies it essentially by the matrix 
$B = ((I + e^{i A m})/2)^p$.
After this "filtering" step, it evaluates whether or not the resulting vector is close to being an eigenvector of $A$,
and keep it if it is.
To understand the behavior of the algorithm, it is insightful to consider the behavior
in the eigenbasis of $A$.
$$
w = \sum_i \alpha_i w_i,
$$
where $\{w_i\}_{i\in [n]}$ is an orthonormal basis for $A$ corresponding to eigenvalues $\{\lambda_i\}_{i\in [n]}$.
If $\{ m \lambda_i\}$, i.e. - the fractional part of $m \lambda_i$, is very close to $0$, and $\{ m \lambda_j\}$ is $\sim 2\ln n/p$ far from $0$ for all $j \neq i$,
then after multiplication by $B$ and normalization, all eigenvectors $w_j$ for $j\neq i$
are attenuated by factor $1/n^2$ relative to $w_i$, and hence
the resulting vector is $1/n$ close to an {\it eigenvector} of $\lambda_i$.

Hence, a sufficient condition on the number $m$ that would imply that $w = Filter(A,m,\delta)$
is an approximation of the $i$-th eigenvector is the following property:
$\{ m \lambda_i \}$ is very close to $0$, and for all $j\neq i$
$\{m \lambda_j\}$ is bounded away from $0$.
This corresponds to the fact that $m$ {\it separates} $i$ in $A$, as assumed.

So to sample uniformly an approximate eigenvector, we would like to call $Filter(A,m,\delta)$
for $m\sim U[M]$ for $M = \poly(n)$ such that
$m$ separates $i$ where $i\sim U[n]$.
The main observation here, is that this condition is satisfied 
if the sequence of residuals of integer multiples of the eigenvalues $S(A)$ defined above has
the aforementioned {\it low discrepancy} property.

Most of the work in this study is devoted to achieving this property.
Computationally, we achieve low-discrepancy of $S(A)$ simply by additive Gaussian perturbation
prior to calling the sampling routine.
We show that if we perturb a matrix using a Gaussian matrix ${\cal E}$ of variance $1/\poly(n)$, then $S(A + {\cal E})$ 
has discrepancy which is $1/\poly(n)$.
Showing this is non-trivial because arbitrary vectors of eigenvalues $\lambda_1,\hdots, \lambda_n$ do not generate low-discrepancy sequences in general
\footnote{Consider for example the sequence of values $1, 1/2,\hdots, 1/n$.  Then multiplying these numbers individually by a random integer $m$ and taking the residual would map
to $0$ all values $1/i$ for which $i | m$, and there is a logarithmic number of these on average.
This sequence of eigenvalues is well-separated, and at least potentially could arise from a random matrix.
We show, however, that this is not the typical case.}
, and 
on the other hand we are also severely limited in our ability to perturb the eigenvalues without deviating too much from the original matrix.
This is the subject of our main technical theorem \ref{thm:sep},
which may be of independent interest:
\begin{theorem*}(sketch)
Let $A$ be an $n\times n$ Hermitian matrix, and ${\cal E}$ be a standard Gaussian matrix.
For any $a>0, b>0$ there exists $M = M(a,b) = \poly(n)$ such that w.p. at least $1 - n^{-b}$
the sequence of residuals of eigenvalue multiples of $A + n^{-a} \cdot {\cal E}$ of length $M$
has discrepancy at most $n^{-b}$.
\end{theorem*}

Perturbing the input matrix has the additional benefit of making sure that $A$ has a exactly $n$ unique eigenvalues
with high probability.  This follows from a breakthrough theorem by Nguyen, Tao and Vu \cite{NTV}
which has provided a resolution of this long-standing open problem, which was considered unproven folklore until that point.
This theorem allows us to handle
general Hermitian matrices without extra conditions on the conditioning number of $A$ / its eigenvalue spacing.

\subsection{Prior Art}
 
There are numerous algorithms for computing the $\ASD$ of a matrix, relying most prominently on the QR decomposition
\cite{Trefethen}.
For specific types of matrices, like tridiagonal matrices much faster algorithms are known \cite{Reif}, but here we consider the most general Hermitian case.
We summarize the state of the art algorithms for this problems in terms of their complexity (boolean / arithmetic, serial / parallel)
and compare them to our own:

\begin{center}
    \begin{tabular}{ | p{2cm} | p{2cm} | p{2cm} | l | l | p{2.5cm} |}
    \hline
                     & Arithmetic Complexity 	& Boolean Complexity 	& Parallel Time 		 	& Log-Stable  & Comments \\ \hline
    Csanky     & $\tilde{O}(n^{\omega+1})$ 	& $\tilde{O}(n^{\omega+2})$ 	& $\log^2(n)$ 		& NO  & \\ \hline
    Demmel et al. \cite{DDH07}  & $\tilde{O}({n^{\omega}})$ 		& $\tilde{O}({n^{\omega}}){ (*)}$    &  		N/A	
    & YES & \tiny{${*}$ Conjectured by us, by modifying the algorithm.}
      \\ \hline
     Bini et al., Reif \cite{Bini92, Reif} & $\tilde{O}(n^{\omega})$ & $\tilde{O}(n^{\omega+1})$ & $O(\log^2(n))$ & NO &
     \tiny{Working with $\Omega(n)$ bit Integers}\\ \hline
     New	    & $\tilde{O}(n^{\omega+1})$		& $\tilde{O}(n^{\omega+1})$ 	& $\log^2(n)$	& YES	&  \\ \hline						
     \end{tabular}
\end{center}

Comparing our algorithm to the best known $\NC^{(2)}$ algorithms,
it is more efficient by a factor of $n$ compared with Csanky's algorithm \cite{Kozen}.
Notably, our algorithm is completely disjoint from Csanky's techniques - which rely on computing explicitly high powers of the input matrix,
and computes the characteristic polynomial of the matrix using the Newton identities on the traces of those powers.
This is an inherently unstable algorithm as it finds the eigenvalues by approximating the roots of the characteristic polynomial and
small perturbation to the coefficients of the polynomial may lead to large deviations of the roots.

The algorithms of Demmel et al., Bini et al. and Reif, rely on efficient implementation of variants of the QR algorithm. Our 
asymptotic bounds are worse then Demmel et al.  in terms of total arithmetic/boolean complexity, though we conjecture that this is an artifact of our proof strategy, and not
an inherent problem (see the section on open problems), and in fact, a variant of the algorithm
could probably achieve a boolean complexity of $O(n^{\omega})$.
We note that the QR algorithm is not known to be parallelizable in a stable way, and hence the fast parallel algorithms
of Bini et al. and Reif are  not stable and probably impractical.
In fact the QR decomposition
has been shown, for standard implementations like the Given's or Householder method, to be $P$-complete \cite{LMM} assuming the real-RAM model.
Thus, it is unlikely to be stably-parallelizable unless ${\rm P} = \NC$.
\footnote{We point out that the algorithm of Reif \cite{Reif} achieves a QR factorization in parallel time $O(\log^2(n))$ in the arithmetic model, thus showing that QR is indeed parallelizable, but it relies on computations modulo large integers and therefore not stable and not practical.}


Thus, to the best our knowledge, our algorithm is the first parallel algorithm for the $\ASD$ of general Hermitian matrices 
that is both parallel and stable.
In particular it achieves the smallest bit-complexity of any $\RNC^{(2)}$ algorithm to date.
We conjecture that our approach may present a practical and parallel alternative to computing the $\ASD$.
We dwell on this point a bit more in Section \ref{sec:numerics}.

\subsubsection{Comparison to the power method / QR algorithm}
An arguably natural benchmark by which to test the novelty of the proposed algorithm is the
iterative power-method for computing the eigenvalues of a Hermitian matrix.
In this method, one starts from some random vector $b_0$, and at each iteration $k$ sets:
$$
b_{k+1} = \frac{A b_k}{\left\|A b_k \right\|}.
$$
This method can achieve polynomially good approximation of the top eigenvalue in time which is logarithmic in $n$,
for $A$ with a constant spectral-gap.

Both the power method and our proposed scheme are similar in the sense that they attempt to extract the eigenvectors of the input matrix directly.
Also, if two eigenvalues are $\eps$-close in magnitude, for some $\eps>0$, then
they require essentially the same exponent of $A$ in the power method, and of $e^{2\pi iA}$ in our
scheme to distinguish between them.
However, the similarity stops here.
We maintain, that the power method is both conceptually different, and for general Hermitian matrices performs
much worse, in terms of running time, compared with our proposed algorithm.

Conceptually, in the power method, we seek to leverage the difference in {\it magnitude} between adjacent
eigenvalues in order to extract the eigenvectors.
On the other hand, in our proposed scheme we 
recast the problem on the unit sphere $S^{(1)}$, where we are interested in the spacing of
the residuals of integer multiples of the eigenvalues.
Worded differently, our setting exploits the additive group structure of the eigenvalues modulo $1$,
whereas the power method distinguishes between them multiplicatively.
In the additive group setting, the advantage is that we can consider the discrepancy of the
sequence of residuals, and analyze how quickly these residuals mimic a completely independent random distribution.
Furthermore, in the additive setting there is inherent symmetry between the eigenvalues,
as no eigenvalue is more likely to be sampled than another.
This allows for a natural parallelization of the algorithm to extract simultaneously approximation of all eigenvectors.

The well-known QR algorithm for eigendecomposition \cite{Golub} is the de-facto
standard for computing the $\ASD$, and is, in a sense, a parallel version
of the power-method.
That algorithm applies an iterated sequence of of QR decompositions:
At each step $k$ we compute (where $A_1 = A$ - the input matrix)
$$
A_k = Q_k R_k,
$$
and then
set 
$$
A_{k+1} = R_k Q_k.
$$
The algorithm runs in time $\tilde{O}(n^3)$, by applying several pre-processing steps \cite{Golub},  and the fast variant of Demmel et al. in time $O(n^\omega)$.
However, as stated above, the QR decomposition which is at the core of these methods is not known to be stably parallel.

\subsection{Open Questions}

We outline several open questions that may be interesting to research following this work:
\begin{enumerate}
\item
Is it possible to attain a serial run-time of $O(n^{\omega})$ for this algorithm?
We conjecture that this is possible based on numerical evidence for a variant of this algorithm, yet we do not have a proof of this fact.
While not directly improving on the best previously known serial run-time (\cite{DDH07}), it would reduce the overall work
performed by our parallel $\RNC^{(2)}$ algorithm to match the work done by state-of-the-art serial run-time algorithm.
\item
What other linear-algebra algorithms can be designed using our methods ?
We would like these algorithms to improve on previous algorithms in either the stability, boolean complexity, parallel run-time, 
or all these parameters together.
\item
Could one reduce the number of random bits required by the algorithm?
Currently - we show that using $\tilde{O}(n^2)$ random bits - i.e. applying additive Gaussian perturbation results
in a matrix whose eigenvalues seed a low-discrepancy sequence.
However, can one do away with only $\tilde{O}(n)$ random bits - by applying a tri-diagonal perturbation to the matrix?
\item
Is our algorithm practical ?  What is the actual run-time of the algorithm on matrices of "reasonable" size,
and does it compare with state-of-the art parallel algorithms?
Our numerical evidence suggests that in practice our algorithm may run much faster than the analytical asymptotic
bounds we provide here.
\end{enumerate}

\newpage
\section*{Table of contents}
\renewcommand*{\contentsname}{}
\tableofcontents

\newpage
\section{Preliminaries}

\subsection{Notation}

A random variable $x$ distributed according to distribution ${\cal D}$ is denoted by $x\sim {\cal D}$.
For a matrix $X$, $\|X\|$ signifies the operator norm of $X$.
For a set $S$, $U[S]$ is the uniform distribution on $S$.
For integer $M>0$ the set $[M]$ is the set of integers $\{0,1,\hdots, M-1\}$.
For real number $x$, $\{x\}$ denotes the fractional part of $x$: $\{x\} = x - \floor{x}$.
For a Hermitian $n\times n$ matrix $A$,  with eigenvalues $\{\lambda_i\}_{i=1}^n$, $\lambda_1\geq \lambda_2 \geq \hdots \geq \lambda_n$
${\cal L}(A) = (\lambda_1,\hdots,\lambda_n)\in \R^n$ denotes the vector of sorted eigenvalues of $A$.
For a measurable subset $S\subseteq \R^n$ $\vol(S)$ denotes the volume of $S$.
$\Phi$ is the empty set.
GUE is the global unitary ensemble of random matrices.
$\N,\Z,\C$ signify the natural, integer, and complex numbers, respectively.
For a matrix $A$, $A^H$ is the Hermitian conjugate-transpose of $A$.
For number $n>0$ $\ln n$ denotes the natural logarithm, and $\log n$ denotes the binary logarithm.
$\mu(\eta,\sigma^2)$ is the standard Gaussian measure with mean $\eta$ and variance $\sigma^2$.
$U(n)$ is the set of $n\times n$ unitary matrices.

\subsection{Definitions}

\subsubsection{Hermitian matrices}

We repeat again the definition of the $\ASD$ due to its importance:
\begin{definition}

\textbf{Approximate Spectral Decomposition} - ${\rm ASD}(A,\delta)$

\noindent
Let $A$ be some $n\times n$ $\delta$-separated Hermitian matrix.
An approximate spectral decomposition of $A$, with accuracy parameter $\delta = 1/\poly(n)$ is a set of unit vectors $\{v_i\}_{i=1}^n, \|v_i\| = 1$ such that
each $v_i$ has $\|A v_i - \lambda_i v_i \|_2 \leq \delta$, for some number $\lambda_i\in \C$,
and
$$
\left\| \sum_{i\in [n]} \lambda_i v_i v_i^T - A \right\| \leq \delta,
$$
where $\|X\|$ is the operator norm of $X$.
\end{definition}
By standard arguments the above can be generalized to an arbitrary $n \times n$ matrix $A$, by considering the Hermitian matrix $A^H A$, in which case
${\rm ASD}(A^H A,\delta)$ is an approximation of the {\it singular vectors} (and singular values) of $A$.
We note that if $\{v_i\}_{i\in [n]}$ is an $\ASD(A,\delta)$ for some matrix $A$, and $B$ is a matrix such that $\|A - B\| \leq \delta$,
then by the triangle inequality $\{v_i\}_{i\in [n]}$ is an $\ASD(A, \sqrt{2} \delta)$.

\subsubsection{Complexity}

\begin{definition}
Let $\omega$ denote the infimum over all $t$ such that any two $n\times n$ matrices can be multiplied using a number of 
products at most $n^t$. 
\end{definition}
\noindent
\noindent
The current best upper-bound on $\omega$ is $2.372$ due to Williams \cite{Williams}.

\begin{definition}\label{def:nc}

\textbf{Class NC}

\noindent
The class ${\rm NC}^{(k)}$ is the set of problems computed by uniform boolean circuits, with a polynomial number of gates, and depth at most $O(\log^k n)$.

\end{definition}

\noindent
We will require the following known fact:
\begin{fact}\label{fact:sort}\cite{AKL}
There exists an algorithm for sorting $n$ numbers in time $\log(n)$, using $n$ processors.
\end{fact}

\begin{definition}

\textbf{Class RNC}

\noindent
The class ${\rm RNC}^{(k)}$ is the set of problems that can be computed by uniform boolean circuits, with a polynomial number of gates,
accepting a polynomial number of random bits, and depth at most $O(\log^k n)$.
\end{definition}
For simplicity, we shall assume in this work that $\RNC$ circuits are allowed to accept $t$-bit numbers, sampled from a truncated
Gaussian distribution, and discretized to $t$-bits of precision.

\subsubsection{Stable Computation}

Following Demmel et al. \cite{DDH07} we define the notion of log-stability as one
where truncating each binary arithmetic operation to $O(\log(n))$ bits of precision
doesn't change the result by much:

\begin{definition}

\textbf{$(t,\delta)$-stable randomized computation}

\noindent
Let $C$ denote a randomized arithmetic circuit, 
and ${\cal D}$ be its output distribution supported on $\R^n$.
Let $D$
denote the discretization of $C$ to $t$ bits as follows: each infinite-precision arithmetic operation is followed by rounding to $t$ bits.
Let ${\cal D}'$ denote the output distribution of $D$.
$C$ is said to be $(t,\delta)$-stable if
$$
\forall x \ \ \exists y, \ \
{\cal D}(x) = {\cal D}'(y) \mbox{ and } 
\|x - y\| \leq \delta.
$$
\end{definition}

\begin{definition}

\textbf{Log-stable computation}

\noindent
Let $C$ be a randomized arithmetic circuit that accepts $n$ input numbers.
$C$ is said to be log-stable if for any $\delta = 1/\poly(n)$ it is $(t,\delta)$-stable for some $t = O(\log(1/\delta))$.
\end{definition}

\section{Additive Perturbation}

Matrix perturbation is a well-developed theory \cite{Stewart,Golub} examining
the behavior of eigen-values and eigen-vectors under additive perturbation, usually
much smaller compared to the norm of the original matrix.
While general eigenvalue problems are usually unstable against perturbation,
for Hermitian matrices the situation is much better:
the Bauer-Fike theorem states that the perturbed eigenvalues can only deviate 
from the original eigenvalues by an amount corresponding to the relative strength
of the perturbation.

In particular, when the perturbed matrix $A$ is $\delta$-separated one can compute
an explicit estimate for the behavior of the perturbed eigenvalues.
We use here a quantitative estimate by \cite{Stewart}:
\begin{fact}\label{fact:Stewart}
\textbf{Rayleigh quotient for well-separated eigenvalues}

\noindent
Let $A$ be a $\delta$-separated $n\times n$ Hermitian matrix with eigenvalues $\lambda_1> \lambda_2 > \hdots > \lambda_n$,
and corresponding orthonormal basis $\{v_i\}_{\in [n]}$.
Let ${\cal E}$ be an additive perturbation of $A$ satisfying $|{\cal E}_{i,j}| \leq \eps$ for all $i,j$.
Let $\tilde{\lambda_i}$ denote the $i$-th eigenvalue of $A + {\cal E}$.
There exists a constant $c>0$ satisfying:
$$
\forall i\in [n] \ \ \tilde{\lambda_i} = \lambda_i + v_i^H {\cal E} v_i + \zeta_i,  \quad |\zeta_i| \leq c \eps^2/ \delta.
$$
\end{fact}

\noindent
In fact, if the perturbation ${\cal E}$ is GUE a stronger characterization is readily available:
\begin{corollary}\label{cor:gaus1}
Let $A$ be a $\delta$-separated $n\times n$ Hermitian matrix with eigenvalues $\{\lambda_i\}_{i\in [n]}$,
and corresponding orthonormal basis $\{v_i\}_{\in [n]}$.
Let ${\cal E}$ be GUE.
Then the eigenvalues $\{\lambda_i'\}_{i\in [n]}$ of the perturbed matrix $A' = A + \eps \cdot {\cal E}$ are
distributed as follows:
$$
\forall i\in [n] \ \ \tilde{\lambda_i} = 
(1 - \alpha) \cdot \mu(\lambda_i,\eps^2) + \alpha \cdot {\cal D} + \zeta_i, \quad
0 \leq \alpha \leq 2^{-n}, |\zeta_i| \leq 16 c n \cdot \eps^2 / \delta
$$
for some distribution ${\cal D}$.
\end{corollary}
\begin{proof}
By Fact \ref{fact:Stewart} the eigenvalues $\lambda_i'$ behave as
$$
\lambda_i' = \lambda_i + v_i^H {\cal E} v_i + \zeta_i,  \quad |\zeta_i| \leq c \max_{i,j} |{\cal E}_{i,j}|^2/\delta.
$$
The standard Gaussian satisfies:
$$
\P_{\mu} \left( |x| \geq 4 \sqrt{n} \right) \leq 2^{-2n}.
$$
Thus, by the union bound we have that $|{\cal E}_{i,j}| \leq 4 \sqrt{n}$ for all $i,j$ w.p. at least $1 - 2^{-n}$.
Hence, w.p. at least $1- 2^{-n}$ we have: 
$$
\forall i\in [n] \quad
|\zeta_i| \leq  c\max_{i,j} |{\cal E}_{i,j}|^2/\delta \leq 16 c n \cdot \eps^2/ \delta,
$$
Suppose that this is the case, and let ${\cal E}'$ denote the GUE matrix, {\it conditioned} on having
bounded entries.
We can write:
$$
{\cal E}' = (1-\alpha) \cdot {\cal E} + \alpha \cdot {\cal D},
$$
where ${\cal D}$ is some distribution on $n\times n$ matrices 
and $0\leq \alpha \leq 2^{-n}$.
The distribution ${\cal E}$ is invariant under unitary conjugation, i.e.: 
$$
\forall u\in U(n) \quad
u {\cal E} u^H = {\cal E}
$$
then
$$
\forall u\in U(n) \quad
u {\cal E}' u^H = (1-\alpha) \cdot {\cal E} + \alpha \cdot {\cal D}',
$$
where
${\cal D}'  = u {\cal D} u^H$ is some distribution on $n\times n$ matrices.
Hence, up to statistical distance at most $\alpha$ we can assume w.l.o.g. that $v_i = e_i$.
This implies that
$$
\lambda_i' = 
\lambda_i + 
(1 - \alpha) \cdot {\cal E}_{i,i} + 
\alpha \cdot {\cal D}' + 
\zeta_i,
$$ where
$|\zeta_i| \leq 16 c n\cdot \eps^2/ \delta$.
Since ${\cal E}_{i,i} = \mu(0,\eps^2)$ then
$$
\lambda_i' =
(1-\alpha) \cdot \mu(\lambda_i, \eps^2) + \alpha \cdot {\cal D}' + \zeta_i,
\quad
\quad |\zeta_i| \leq 16 c n\cdot \eps^2/ \delta.
$$

\end{proof}

The stability of eigenvalues has also been generalized to eigenvectors  
and even general invariant subspaces \cite{Stewart} in the following sense: if there is
a cluster of eigenvalues that is "well-separated" from all other eigenvalues,
then the orthogonal projection onto the subspace spanned by the corresponding eigenvectors
is likewise stable under additive perturbation whose scale is negligible compared to the separation
of the eigenvalue of this cluster from the rest of the spectrum.
In particular, if a matrix is $\delta$-separated, then its eigenvectors are individually stable as follows:
\begin{fact}\label{fact:Stewart2} \cite{Stewart}
\textbf{Rayleigh quotient for well-separated invariant subspaces}

\noindent
Let $A$ be a $\delta$-separated $n\times n$ Hermitian matrix with eigenbasis $\{v_i\}_{i\in [n]}$, and
$$
A_1 = A + {\cal E}, \quad \|{\cal E}\| \leq \eps.
$$
There exists an orthonormal basis $\{v_i' \}_{i\in [n]}$ and a constant $c>0$ satisfying
$$
\forall i\in [n] \quad
\| v_i - v_i' \| \leq c \eps^2/ \delta.
$$
\end{fact}

Our interest in additive perturbation, however, is not confined just to "stability" arguments.
In fact, our main reason for using perturbation is to cause a scattering of the eigenvalues.
The first step of our algorithm in fact applies additive perturbation to provide a minimal spacing between eigenvalues.
Recently Nguyen et al. \cite{NTV} have provided the first proof that applying additive perturbation
to any Hermitian matrix using
a so-called Wigner ensemble, an ensemble of random matrices that generalize GUE, in fact
causes the eigenvalues of the perturbed matrix to achieve a minimal inverse polynomial separation.
We state their result:
\begin{lemma}\label{lem:NTV}\cite{NTV}
\textbf{Minimal eigenvalue spacing}

\noindent
Let $M_n = F_n + \eps \cdot X_n$, where $F_n$ is a real symmetric matrix, $\left\|F_n\right\|_2 \leq 1$, $\eps = n^{-\gamma}$ for some constant $\gamma>0$,
and $X_n$ is GUE.
Let $\lambda_1\geq \lambda_2 \geq \hdots \geq \lambda_n$ denote the eigenvalues of $M_n$, and put $\alpha_i = \lambda_i - \lambda_{i+1}$ for all $i<n$.
Then for any fixed $A>0$ there exists $B = B(\eps) >0$, such that
$$
\max_{i\in [n]} \P \left(\alpha_i \leq n^{-B}\right) = {\rm O}(n^{-A}).
$$
In particular \footnote{applying the union bound over all eigenvalues} for any $A>0$ there exists $B>0$ such that 
$$
\P \left(\min_{i\in [n]} \alpha_i \geq n^{-B}\right) = 1 - {\rm O}(n^{-A}).
$$ 
\end{lemma}

\noindent
Using the lemma above we define:
\begin{definition}\label{def:bstar}
For any $\delta = 1/\poly(n)$, let $B^*(\delta)$ denote the largest number $B>0$ such that
for every $F_n$ the matrix $M_n = F_n + \delta X_n$ satisfies:
$$
\P \left(\min_{i\in [n]} \alpha_i \geq n^{-B}\right) \geq 0.99
$$
\end{definition}

\section{Low-Discrepancy Sequences}\label{sec:random}

\subsection{Basic Introduction}

Low discrepancy sequences (or "quasi-random" sequences) are a powerful tool in random sampling methods.  
Roughly speaking, these are deterministic sequences that
visit any "reasonable" subset $B$ a number of times that is roughly proportional to the volume of $B$,
up to some small additive error, called the discrepancy.
\begin{definition}\label{def:discrepancy}

\textbf{Multi-dimensional discrepancy}

\noindent
For integer $s$, put $I^s = [0,1)^s$.
Given a sequence $x = (x_n)_{n=1}^{N}$, with $x_n \in I^s$ the discrepancy $D_N(x)$ is defined as:
$$
D_N(x) = \sup_{B\in {\cal B}} \left\{ \left| \frac{1}{N}\sum_{n=1}^N \chi_B(x_n) - {\rm vol}(B) \right| \right\},
$$
where $\chi_B(x_n)$ is an indicator function which is $1$ if $x_n\in B$ and $0$ otherwise,
and ${\cal B}$ is a non-empty family of Lebesgue-measurable subsets of $I^s$.
\end{definition}
In this work, we shall define ${\cal B}$ as the set of all $s$-products of intervals 
$$
\prod_{i=1}^s [u_i,v_i], \quad [u_i,v_i] ({\rm mod } 1)\subseteq [0,1).
$$
Often in literature, one considers instead the star-discrepancy $D_N^*(x)$ which
is defined by optimizing over the set of all intervals of the form $\prod_{i=1}^s [0,u_i]$.

The definition of discrepancy naturally admits an interpretation in terms of probability:
a sequence $x = \{x_n\}_{n\in [N]}$ can be interpreted as a random variable $x$,
that takes the value $x_n$, as $n$ is the random variable $n\sim U[N]$.
Saying that $D_N(x) \leq D_N$ means that for any $B\in {\cal B}$ the probability that $x_n$ is contained in $B$ is equal to $\vol(B) + \eps$,
where $|\eps| \leq D_N$.
In this work, deviating somewhat from standard terminology,  we shall often refer to the discrepancy of a {\it distribution} $x$ on length $N$, $s$-dimensional sequences.
In this case $D_N(x)$ will denote 
$$
D_N(x) = \sup_{B\in {\cal B}} 
\left\{ \left|\P_{x'\sim x,n\sim U[N]}(x_n' \in B) - {\rm vol}(B) \right| \right\},
$$

Low-discrepancy sequences have much in common with random sampling, or the Monte-Carlo method, in the sense
that they visit each cube a number of time that is roughly proportional to its volume, up to a small additive error.
Yet, contrary to the Monte-Carlo method, such sequences are {\it not} random, but only appear to be random
in the sense above.
Arguably such sequences are most useful in the context of numerical integration: instead of computing
an integral of a continuous function $f$ over the $s$-dimensional cube $[0,1)^s$ one can replace
it with the average value of the function over the points of a low-discrepancy sequence $x= \{x_i\}_{i=1}^N$.
The well-known Koksma-Hlawka theorem, then connects the approximation error as a function
of the discrepancy $D_N(x)$ as follows:
$$
\left|
\frac{1}{N} \sum_{i=1}^N f(x_i) - \int_{[0,1)^s} f(z) dz
\right|
\leq
V(f) \cdot D_N^*(x),
$$
where $V(f)$ is the bounded-variation of $f$ on $[0,1)^s$.

There are deterministic $s$-dimensional sequences $x = \{x_i\}_{i=1}^N$ with discrepancy as low as 
$$
D_N(x) \leq C \cdot \frac{\log^s N}{N},
$$
and matching lower-bounds (up to constant factors) on the smallest possible discrepancy are known for $s=1$
\cite{Nied}.
Hence, usually one considers low-discrepancy sequences that are very long ($N$) compared to the dimension ($s$).

At this point, it may be insightful to consider an example of a low-discrepancy sequence:
the well-known van-der Corput sequence \cite{Nied}:
consider the binary expansion of a positive integer 
$$
\forall n\in [N = 2^b] \quad
n = \sum_{i=0}^b \alpha_i 2^i,
$$
then the van-der Corput sequence $x = \{x_n\}_{n=1}^N$ is defined as:
$$
\forall n\in [N = 2^b] \quad
x_n = \sum_{i=0}^b \alpha_i 2^{-i-1}.
$$
This sequence has 
$$
D_N^*(x) \leq C \frac{\log N}{N}.
$$

We note that the discrepancy upper-bound decays asymptotically like $O(1/N)$ (assuming small dimension $s$)
whereas the Monte-Carlo method converges more slowly, behaving as $O(1/\sqrt{N})$ - and
hence these sequences are often preferred as a method of numerical integration to Monte-Carlo.
They are also advantageous compared with purely deterministic methods, like defining a fine-resolution grid,
because usually one can increase the length of the quasi-random sequence, and reduce its discrepancy
while making use of all previous points of the sequence.

\subsection{Some basic facts}

We require a Lemma [2.5] due to Niederreiter \cite{Nied}.
\begin{lemma}\cite{Nied}\label{lem:nied}
\textbf{Small point-wise distance implies similar discrepancy}

\noindent
Let $x_1,\hdots, x_N$, $y_1,\hdots, y_N$ denote two $s$-dimensional sequences for which $|x_{n,i} - y_{n,i} | \leq \eps$, for all $n\in [N],i\in [s]$.
Then the discrepancies of these sequences are related by:
\be
\left|
D_N(x_1,\hdots,x_N) - D_N(y_1,\hdots,y_N)
\right|
\leq
s \cdot \eps.
\ee
\end{lemma}

\noindent
We prove an additional fact:
\begin{fact}\label{fact:conv}
Let $x = \{x_n\}_{n\in [N]}$ be a distribution on sequences with discrepancy at most $D_N(x)$, and let 
$y = \{y_n\}_{n\in [N]}$ denote the following sequence:
$$
y_n = x_n + z_n,
$$
where $z = \{z_n\}_{n\in [N]}$ is some sequence chosen independently from $x$.
Then
$$
D_N(y) = D_N(x).
$$
\end{fact}
\begin{proof}
For each $S \in {\cal B}$ we have
$$
\P_{y,n\sim U[N]}(y_n\in S) = 
\P_{x,z,n\sim U[N]}(x_n + z_n\in S) =
\int_{[0,1)^s} \P_{x,n\sim U[N]}(x_n\in S - z_n) \cdot \P_{z,n\sim U[N]}(z_n) d z_n
$$
By the discrepancy assumption on the sequence $x$ we have that
$$
\forall S\in {\cal B} \quad 
\left|
\P_{x,n\sim U[N]}(x_n\in S) - \vol(S) \right|
\leq
D_N(x).
$$
Therefore by the above
$$
\P_{y,n\sim U[N]}(y_n\in S)
= \int_{[0,1)^s} 
 ( \vol(S) + \eps(z_n)) \cdot \P_{z,n}(z_n) d z_n,   \ \  |\eps(z_n)| \leq D_N(x).
$$
and hence by convexity:
$$
\P_{y,n\sim U[N]}(y_n\in S) = \vol(S) + \eps, \quad |\eps| \leq D_N(x),
$$
which implies the claim by the definition of discrepancy.

\end{proof}

\subsection{The Good Seed Problem}

In our context we will be interested in sequences $x = \{x_n\}_{n=1}^N$ of the form
$$
x_n = \left\{  \frac{g n}{N} \right\},
$$
where $g\in [N]^s$ is some $s$-dimensional vector, called the {\it seed} of the sequence.
Specifically, the vector $g/N$ will be the vector of eigenvalues of an $n\times n$ Hermitian matrix $A$
whose spectrum $g = {\cal L}(A)$ we would like to analyze.
Since it is unreasonable to assume that the input matrix has a spectrum that is
a good seed,
we would like to find a perturbation of the matrix $A' = A + {\cal E}$ such that
$g'  = {\cal L}(A')$ has a corresponding sequence, defined as above, with
low-discrepancy.

Niederreiter has shown \cite{Nied} that if $g$ is sampled uniformly on $[N]^s$ then it is
a good seed with high probability:
\begin{lemma}
Let $s,N$ be an integers and $g\sim U([N]^s)$.
Then
$$
\P \left( D_N(x) \leq \frac{\log^s N}{N} \right) \geq 1 - 1/N.
$$
\end{lemma}
For our application we require that $N = \poly(n)$, and $s=2$, in which
case the above discrepancy is sufficiently low for our purposes.
Yet, since it requires the normalized seed $g/N$ to be essentially uniform on $[0,1)^n$, it implies
that the corresponding matrix perturbation ${\cal E}$ added to $A$ must be very strong - thereby loosing
all connection to the input matrix.

\subsection{Finding Reasonably-Good Seeds Locally}

To bridge the gap between weak-perturbation and low-discrepancy
we show a new lemma, which may be of independent interest: 
it allows to trade-off the extent to which 
$g$ is random, and the discrepancy of the sequence generated by $g$.
Specifically, we will show that if $g/N$ is uniform on ${\it cubes}$ of much smaller side-length, i.e.
at least $1/\sqrt{N}$, then the resulting sequence has discrepancy $O(\log^s N / \sqrt{N})$.
This is the subject of the following lemma:

\begin{lemma}\label{lem:discrepancy}

\noindent
We are given integer $N$, 
with prime divisor $M = \Theta(N^a)$ for some constant $a>0$,
and an integer $s$.
Let $g = (g_1,\hdots,g_s)\in N^s$, such that each coordinate $g_i$ is independently
chosen uniformly on some interval $I_i\subseteq [N]$ of size $M$.
Let $x = x(g) = \{x_n\}_{n=1}^N$ be the following $s$-dimensional sequence of length $N$
corresponding to residuals of $g$: 
$$
x_n = 
\left\{ \frac{g\cdot n}{N} \right\}.
$$
Then
$$
\P_g
\left(
D_N(x)
\leq
2 \log^s(M) / \sqrt{M}
\right)
\geq
1 - 1/\sqrt{M}.
$$
\end{lemma}

\begin{proof}
For integers $P,s$ put $C_s^*(P)$ as the set of all vectors in $\Z^s$ with entries in $[-P/2,P/2)\cap \Z$,
excluding the all-zero vector.
Following Niederreiter \cite{Nied} we define for each $h\in \Z^s$
\begin{equation}
r(h) = \prod_{i=1}^s r(h_i),
\end{equation}
where $r(h_i) = \max(1,|h_i|)$.
For $g=(g_1,\hdots,g_s)\in \Z^s$, we denote:
\begin{equation}
R(g,P) = \sum_{h\cdot g = 0 (mod P), h\in C_s^*(P)} r(h)^{-1}.
\end{equation}
By Theorem [5.10] of \cite{Nied} when each $g_i$ is randomly chosen on the entire interval $[P]$, for prime $P$,
then
\begin{equation}
\mathbf{E}_{g}\left[ R(g,P) \right] = O \left(\frac{\log^s(P)}{P}\right).
\end{equation}
Since $R(g,P)\geq 0$ for all vectors $g$, then 
\begin{equation}
\P_g \left(R(g,P) \geq \frac{\log^s(P)}{\sqrt{P}}\right) \leq 1/\sqrt{P}.
\end{equation}
or
\begin{equation}\label{eq:alluni}
\P_g
\left( 
R(g,P)\leq \frac{\log^s(P)}{\sqrt{P}}
\right) 
\geq 1-1/\sqrt{P}.
\end{equation}
Let us use the above equation to upper-bound the discrepancy of $S(g)$.
Recall that $M$ is a prime divisor of $N$, with $M = \Theta(N^{a})$.
We first observe that:
\begin{equation}
R(g,N) = 
\sum_{h\cdot g = 0 (mod N), h \in C_s^*(N)} r(h)^{-1}
\leq
\end{equation}
\begin{equation}
\sum_{h\cdot g = 0 ({\rm mod } M), h \in C_s^*(M)} r(h)^{-1}
+
\sum_{h\cdot g = 0 ({\rm mod } M), h \in C_s^*(N), \exists i, \mbox{ s.t. } |h_i|\geq M} r(h)^{-1}
\end{equation}
We note that the second term is at most
\begin{equation}
\frac{s}{M} \sum_{h\in \Z^{s-1}} r(h)^{-1} = \frac{s\cdot \log^{s-1}(N)}{M}
\end{equation}
Regarding the first term, $h\cdot g = 0 ({\rm mod } M)$ if and only if $h\cdot (g({\rm mod } M)) = 0 ({\rm mod } M)$,
and since each $g_i$ is uniform on some interval of size $M$, 
then $g_i({\rm mod } M)$ is uniform on $[M]$. 
Since $M$ is prime, we can apply equation (\ref{eq:alluni}) 
to the first term with $P=M$.
Using this equation and the upper-bound on the second term we have:
\begin{equation}\label{eq:upRg}
\P_g
\left(
R(g,N)
\leq
\frac{2\log^s(M)}{\sqrt{M}}
\right)
\geq
1-1/\sqrt{M}.
\end{equation}
According to Theorem 5.6 of \cite{Nied}, the discrepancy of the sequence 
$\left\{g n/N\right\}_{n=0}^{N-1}$ is upper-bounded by:
\begin{equation}\label{eq:Nied}
D_N(S(g))
\leq
\frac{s}{N}
+
2^{-s} \cdot R(g,N).
\end{equation}
Plugging equation (\ref{eq:upRg}) into equation (\ref{eq:Nied}) implies:
\begin{equation}
\P_g
\left(
D_N(S(g)) \leq \frac{2\log^s(M)}{\sqrt{M}}\right) 
\geq 
1-1/\sqrt{M}.
\end{equation}

\end{proof}

\section{A Filtering Algorithm}\label{sec:filter}

In this section we provide the specification of the filtering algorithm, which is the
main computational black box of our algorithm.
This algorithm accepts 
an integer $m$ that separates the $i$-th eigenvalue of a Hermitian matrix $A$
and computes an approximation for the $i$-th eigenvector,
with high probability:
\begin{mdframed}

\begin{algorithm}\textbf{${\rm Filter}(A,m,\delta)$}

\noindent
\begin{enumerate}
\item
Compute parameters:
$$
p = 24 n^2 \ceil{\ln (1/\delta)},
\zeta =  \delta^2 / (2 p m).
$$
\item
\textbf{Sample random unit vector:}

\noindent
Sample a standard complex Gaussian vector $v$, set $w_0 = v / \|v\|$.
\item
\textbf{Approximate matrix exponent: }\label{it:comp1}

\noindent
Compute a $\zeta$ Taylor-series approximation of $e^{i A}$, denoted by $\tilde U$.
\item
\textbf{Raise to power:}

\noindent
Compute $\tilde U^m$ by repeated squaring.
\item
\textbf{Generate matrix polynomial:}

\noindent
Compute $B = \left(\frac{I + \tilde U^m}{2}\right)^p$ by repeated squaring.
\item
\textbf{Filter:}\label{it:comp4}

\noindent
Compute 
$
w
=
\frac{B \cdot w_0}{\| B \cdot w_0\|}.
$
\item
\textbf{Decide:}\label{decide}

\noindent
Set $z = A \cdot w$, $i_0 = \arg\max_{i\in [n]} |w_i|$ and compute
$
c = z_{i_0} / w_{i_0}.
$
If
$$
\left\| A \cdot w - c \cdot w \right\| \leq 3\delta \sqrt{n}
$$
return $w$, and
otherwise reject.
\end{enumerate}
\end{algorithm}
\end{mdframed}


\noindent
\\
We now show that if the algorithm is provided with an integer $m$ that 
separates the $k$-th eigenvalue of $A$ in the sense defined in Definition \ref{def:sep1}, then
the output is close to the $k$-th eigenvector of $A$.
\begin{mdframed}
\begin{theorem}\label{thm:filter}
Let $\delta \leq n^{-10}$ and $\alpha = \sqrt{\ln (1/\delta)}$.
We are given an $n\times n$ Hermitian matrix $A$ with eigenvalues $\{\lambda_i\}_{i\in [n]}$ and 
a corresponding orthonormal eigenbasis $\{v_i\}_{i\in [n]}$.
Additionally, we are provided
an integer $m$ that separates $k$ in $A$, w.r.t.  $B_{in}(\alpha), B_{out}$,
in the sense of Definition \ref{def:sep}.
Let $w = Filter(A,m,\delta)$.
Then
$$
\P \left(
\left\|
w - v_k
\right\| 
\leq \delta
\right)
\geq 1 - 3 n^{-3}.
$$
The algorithm has boolean complexity at most
$
2 c n^{\omega} \cdot \log(2 p^2 m^2 / \delta^2),
$
and runs in parallel time $O(\log^2(n))$.
\end{theorem}
\end{mdframed}

\begin{proof}

Let $\{ \tau_\ell \}_{\ell\in [n]}$ denote the set of eigenvalues of $\tilde U$.
Since $\tilde U$ is a polynomial in $A$ (truncated Taylor series) then $\{v_\ell\}_{\ell\in [n]}$
is also an orthonormal basis for $\tilde U$.
Since in addition $\| \tilde U - e^{i A} \| \leq \zeta$ then
\be\label{eq:tau2}
\forall \ell\in [n] \ \ 
\left|\tau_\ell - e^{i \lambda_\ell}\right| \leq \zeta.
\ee
Let $w' = B \cdot w_0$ and denote
$$
w_0 = \sum_{\ell\in [n]} \beta_\ell v_\ell,
\mbox{\quad and \quad}
w' = \sum_{\ell\in [n]} \alpha_\ell v_\ell.
$$
Since $A,\tilde U$ share the same basis of eigenvectors, then
by the definition of the matrix $B$
the coefficients $\alpha_\ell, \beta_\ell$ are related by: 
$$
|\alpha_\ell|^2
=
|\beta_\ell|^2 \cdot 
\left|
\frac{1 + {\tau_\ell}^{m}}{2}
\right|^{2p}.
$$
So by Equation \ref{eq:tau2}
$$
\frac{
|\alpha_\ell|^2}
{|\beta_\ell|^2}
\geq
\left|
\frac{1 + e^{i m\lambda_\ell}}{2}
\right|^{2p} - 2 p m \zeta.
$$
Since $m$ separates $k$ then 
$\{m \lambda_k\}\in {\cal B}_{in}$,
and for all $\ell\neq k$ we have $\{m \lambda_\ell\}\notin {\cal B}_{out}$.
Thus, for $\ell = k$:
$$
\frac{
|\alpha_k|^2}
{|\beta_k|^2}
\geq
\left|
\left(
1 + \cos(1/ 2n \sqrt{\ln (1/\delta)})
\right) 
\right|^{p} - 2 p m \zeta
$$
Using Claim \ref{cl:cosine}
\be\label{eq:survivor1}
\geq
\left(1 - \frac{1}{4 n^2 \ln (1/\delta)}\right)^p - 2 p m \zeta \geq
\frac{1}{2 e^6}.
\ee
On the other hand, for all $\ell \neq k$ we have:
$$
\frac{
|\alpha_\ell|^2}
{|\beta_\ell|^2}
\leq
\left|
\frac{1 + e^{i m\lambda_\ell}}{2}
\right|^{2p} + 2 p m \zeta.
$$
so since $m$ separates $k$ then:
$$
\leq
\left|
(1 + \cos(1/ 2n)) 
\right|^{p} + 2 p m \zeta
$$
Using Claim \ref{cl:cosine}
\be\label{eq:ratio2}
\leq
(1 - 1/12 n^2)^{24 n^2 \ln (1/\delta)} + 2 p m \zeta
\leq
e^{-2 \ln(1/\delta)} + 2 p m \zeta
\leq
2 \delta^2.
\ee
By Fact \ref{fact:even} for any $\eps = 1/\poly(n)$ there exists a constant $c>0$ such that
$$
\P( \forall i,j \ \ |\beta_j| \leq c |\beta_i| \sqrt{\ln (1/\eps)}/\eps ) \geq 1 - 3n\eps.
$$
Choose $\eps = n^{-4}$.
Then by
Equations \ref{eq:survivor1} and \ref{eq:ratio2}:
$$
\P
\left(
\forall \ell\neq k \quad
\frac{|\alpha_\ell|^2}{|\alpha_k|^2} \leq  c^2 (2\delta^2)   \cdot  (2 e^6) \cdot 4 \ln n \cdot n^8
\right)
\geq 1 - 3 n^{-3}.
$$
and so for $\delta \leq n^{-10}$ there exists $\eta\in \C$, $|\eta| = 1$ such that
$$
\left\|
\frac{w'}{\|w'\|}
-
\eta \cdot v_k
\right\|^2
\leq
\frac{1}{|\alpha_k|^2}
\sum_{j\neq k} |\alpha_j|^2
\leq
16 c^2 n^9 \ln n \delta^2 e^6 < \delta.
$$
for sufficiently large $n$.
Using Claim \ref{cl:u2vec} we conclude that w.p. at least $1 - 3 n^{-3}$ over choices $w_0$, the criterion is met and the algorithm returns a 
vector $w = w' / \|w'\|$ satisfying the equation above.

\item
\textbf{Arithmetic run-time:}
The approximation of $e^{i A}$ by $\tilde U$ requires, using Fact \ref{fact:taylor} a time at most
$$
c n^{\omega} \log(1/\zeta) = c n^{\omega} \cdot \log(2 p m / \delta^2).
$$
Next, the repeated powering of $\tilde U$ to a power $m$
requires time at most:
$
c n^{\omega} \ceil{\log(m)}
$
and the repeated powering of $B$ to the power $p$ requires time at most:
$
c n^{\omega} \ceil{\log(p)}
$
Hence the total complexity is :
$
4 c n^{\omega} \cdot \log(p m / \delta^2)
$.

\item
\textbf{Depth complexity:}
Each matrix product can be carried out in depth $\log(n)$.
Each of steps \ref{it:comp1} to \ref{it:comp4} involves at most $\log(m) + \log(p)$ sequential matrix multiplications.
The sorting algorithm can be computed in depth $O(\log(n))$ by Fact \ref{fact:sort}.
Hence the depth complexity of the entire circuit is at most $\log(n) \cdot (\log(m) + \log(p)) + O(\log(n)) = O(\log^2(n))$.

\noindent
\\
We conclude the proof of the theorem by showing stability:
\begin{claim}\label{cl:stable1}
Under the assumption of Theorem \ref{thm:filter} the algorithm is log-stable.
\end{claim}
\noindent
\textbf{Proof:}
Consider the arithmetic operations involved in computing the filtering algorithm:
\begin{enumerate}
\item
Generating an approximation $\tilde U$ of $e^{i Am}$ as a truncated Taylor series.
\item
Raising $\tilde U$ to a power $m\in [M]$.
\item
Computing $((I + \tilde U)/2)^p$.
\item
Normalizing $B w_0 / \|B w_0\|$.
\end{enumerate}

Consider an arithmetic circuit $C$ implementing the above,
and the circuit $D = D(C,t)$ - the
discretization of $C$ to $t$ bits of precision modeled as follows:
after each arithmetic step, the result is rounded to the nearest
value of $2^{-t}$.
Consider all steps except division.
$A$ is $\delta$-separated so in particular $\|A\| \leq 1$.
Thus, whenever we multiply two matrices at any of the steps above both have norm at most $1$.
Hence, at each rounding step the error is increased by
at most $\sqrt{n} 2^{-t}$.
Finally, considering the final division step, we observe that since $m$ separates $k$,
then by Equation \ref{eq:survivor1} we have $\|B w_0\| \geq 1 / (2 e^6)$.
This implies that the total error is at most $\sqrt{n} (p + M) \cdot 2^{-t} \cdot 2e^6$.
Since $M,p$ are both polynomial in $n$ then for any $\delta = 1/\poly(n)$ the error is at most 
$\delta$ for some $t = O(\log(1/\delta))$.

\end{proof}

\subsection{Supporting Claims}
%
%
%

\begin{claim}\label{cl:u2vec}
Suppose that $\left\| w - v \right\| \leq \delta$ for some unit eigenvector $v$ of $A$, and $\delta \leq 1/4$.
Let $z = A \cdot w$, and
$i_0$ denote $i_0 = \arg\max _{i\in [n]} |w_i|$.
Let $c = z_{i_0} / w_{i_0}$.
Then
$$
\left\| A \cdot w - c \cdot w \right\| \leq 3\delta \sqrt{n}.
$$
\end{claim}

\begin{proof}

Write $w = v + {\cal E}$,  $\|{\cal E}\| \leq \delta$.
Since $\|A\| \leq 1$ then $z = A w = \lambda_v \cdot v + {\cal E}'$,
where $\|{\cal E}'\| \leq \delta$.
Therefore 
$$
z_{i_0} = 
\lambda_v \cdot v_{i_0} + {\cal E}_{i_0}'
=
\lambda_v \cdot (w_{i_0} - {\cal E}_{i_0}) + {\cal E}_{i_0}',
$$
So
$$
z_{i_0} = \lambda_v w_{i_0} + {\cal E}'', |{\cal E}''| \leq 2 \delta.
$$
Since $|w_{i_0}| \geq 1/\sqrt{n}$ then:
$$
c = \frac{z_{i_0}}{w_{i_0}} = \lambda_v + \zeta, |\zeta| \leq 2 \delta \sqrt{n}.
$$
Hence
$$
\left\|
A \cdot w - c \cdot w
\right\|
\leq
\left\|
\lambda_v \cdot v + {\cal E}' - (\lambda_v + \zeta) \cdot 
(v + {\cal E})\right\|
=
\left\| {\cal E}' - (\lambda_v + \zeta) {\cal E} - \zeta v \right\|
$$
which by the triangle inequality is at most
$$
\delta + (1 + 2\delta \sqrt{n}) \delta + 2 \delta \sqrt{n} \leq 3 \delta \sqrt{n},
$$
for sufficiently large $n$.
\end{proof}

\begin{fact}\label{fact:even}

\textbf{Random unit vectors have well-balanced entries}

\noindent
Let $\{v_i\}_{i\in [n]}$ be some orthonormal basis of $\C^n$, $0<\eps  = 1/\poly(n)$,
and $v\in \C^n$ a uniformly random complex unit vector.
For any $i\in [n]$ let $\alpha_i = \left| \langle v,v_i \rangle \right|$.
For any $\eps = 1/\poly(n)$ there exists a number $c_1>0$ independent of $n$, such that
$$
\P
\left( \forall i,j \ \ |\alpha_i|/|\alpha_j|  \leq c_1\sqrt{\ln(1/\eps)}/ \eps 
\right) \geq 1 - 3 n\eps.
$$
\end{fact}

\begin{proof}
Sample a random unit vector by sampling a standard
Gaussian vector and normalizing to unity.
Let $x=(x_1,\hdots,x_n)$ where $x_i\sim \mu(0,1)$, and all are i.i.d.
Since the Gaussian measure is invariant under conjugation by a unitary matrix we can assume
w.l.o.g. that $v_i = e_i$ for each $i$, in which case the $\alpha_i$'s are i.i.d. $\alpha_i \sim \mu(0,1)$.
By the Gaussian measure for every $c_2>0$ there exists $c_1>0$ such that
$$
\P( |\alpha_i| \geq c_1 \sqrt{\cdot \ln n}) \leq n^{-c_2}.
$$
Therefore,
by the union bound
$$
\P( \forall i \ \  |\alpha_i| \geq c_1 \sqrt{\cdot \ln n}) \leq n^{-c_2+1}.
$$
On the other hand, for any $\alpha_i, 0<\eps < 1/n$, we have 
$$
\P(|\alpha_i| \leq \eps) \leq 2\eps.
$$
So by the union bound
$$
\P(\forall i\in [n] \ \ |\alpha_i| \leq \eps) \leq 2n\eps.
$$
Set $c_2 = \ln_n(1/\eps)$.
Since $\eps = n^{-k}$ for some constant $k$, then $c_2 = k$ is also a constant.
Then there exists a constant $c_1$ such that: 
$$
\P \left(\forall i, j \ \ |\alpha_i|/|\alpha_j| \geq c_1 \sqrt{\ln n} /\eps\right)\leq 2n \eps+ n^{-c_2+1}= 3n\eps,
$$
Since normalization does not change these ratios, 
then the maximal ratio between the components of $x$,
is the same as the maximal ratio between the components of $x/\left\|x\right\|$.
This implies the proof.
\end{proof}

\begin{claim}\label{cl:cosine}
\noindent
$
\forall \theta\in [-0.01,0.01]
\quad
1 - \frac{\theta^{2}}{2} \leq \frac{1 + \cos(\theta)}{2} \leq 1 - \frac{\theta^{2}}{3}$.
\end{claim}
\begin{proof}
Follows by truncating the Taylor series of $\cos(x)$ to second order.
\end{proof}

\begin{fact}\label{fact:taylor}

\textbf{Efficient approximation of exponentiated matrix}

\noindent
Given a Hermitian $n\times n$ matrix $A$, $\left\|A\right\| \leq 1$,
and error parameter $\eps >0$,
a Taylor approximation of $e^{iA}$, denoted by $\tilde U_A$ 
can be computed in time $c n^{\omega} \log(1/\eps)$
and satisfies
$\left\|e^{iA} - \tilde U_A\right\|\leq \eps$.
\end{fact}

\begin{proof}
Put $s=\log(1/\eps)$ and consider the Taylor approximation of $e^{iA}$ up to the first $s$ terms:
$$
\tilde U_A :=
\sum_{m=0}^{s-1} 
\frac{(iA)^m}{m!}
$$
Then the approximation error can be bounded as:
\begin{equation}\label{eq:1}
\left\|e^{iA} - \tilde U_A\right\| 
\leq
\sum_{m=s}^{\infty} \frac{\left\|A^m\right\|}{m!}
\leq
\sum_{m=s}^{\infty} \frac{\left\|A\right\|^m}{m!}
\leq
\sum_{m=s}^{\infty} \frac{1}{m!} 
\leq
2^{-s} \leq \eps,
\end{equation}
where we have used the fact that $\left\|A\right\| \leq 1$.
The complexity of the approximation is comprised of $\log(1/\eps)$ matrix products for a total of $c n^{\omega} \log(1/\eps)$.
\end{proof}

\section{Sampling Separating Integers}

In this section we show our main technical tool: which is that 
perturbing a $\delta$-separated Hermitian matrix $A$ by a Gaussian matrix
of a carefully calibrated variance, results in a corresponding
sequence of residuals $S(A)$ having low-discrepancy, at least
for pair-wise variables, which in turn implies that we can
separate each eigenvalue of $A$ almost uniformly:
\noindent
\begin{mdframed}
\begin{theorem}\label{thm:sep}
Let $A$ be a $\delta$-separated $n\times n$ PSD matrix, 
${\cal E}$ GUE, 
$\zeta \leq \min\{\delta^{13}, n^{-50}\}$, and $4 < \alpha \leq n$.
For any $M\geq \zeta^{-1.6}$
we have:
$$
\forall k\in [n] \quad
\P_{{\cal E},m\sim U[M]}
\left(
\mbox{ $m$ separates $k$
in $A + \zeta \cdot {\cal E}$ w.r.t. $B_{in}(\alpha),B_{out}$ }
\right)
\geq
1 / (5 \alpha n)
$$
\end{theorem}
\end{mdframed}

\subsection{Additive Perturbation}

\begin{definition}

\textbf{$(\sigma,\eps)$-normal vector}

\noindent
Let $v = (v_1,\hdots, v_n)$ be a vector of $n$ random variables.
$v$ is said to be $(\sigma,\eps)$-normal, if each $v_i\sim x_i + y_i$ with $x_i\sim \mu(\lambda_i,\sigma^2)$ for
some $\lambda_i \in \R$, 
$x_1,\hdots, x_n$ are independent,
and $|y_i|\leq \sigma \eps$, for all $i$.
\end{definition}

By our definitions above, Gaussian perturbation of a matrix with well-separated eigenvalues
results in a $(\sigma,\eps)$-normal vector as follows:
\begin{fact}\label{fact:perturb}

\textbf{Perturbation of well-separated matrices}

\noindent
Let $A$ be an $n\times n$ $\alpha_{min}$-separated Hermitian matrix with eigenvalues 
$\lambda_1 \geq \lambda_2 \hdots \geq \lambda_n$,
where $\alpha_{min} \geq \eps$, and $\eps>0$ is some constant.
Let ${\cal E}$ be GUE, and $A' = A+\eps^L \cdot {\cal E}$, where $L\geq 2$.
Then w.p. at least $1 - n \cdot 2^{-n}$ 
the vector of eigenvalues of $A'$ $(\lambda_1',\hdots,\lambda_n')$ is $(\eps^L,c n \eps^{L-1})$-normal,
for some constant $c>0$.
\end{fact}

\begin{proof}
We invoke Corollary \ref{cor:gaus1} choosing $\eps$ as $\eps^L$ and $\delta$ as $\eps$.
We get: 
$$
\forall i\in [n] \quad
\lambda_i' = (1 - \alpha) \cdot \mu(\lambda_i,\eps^{2L}) + \alpha \cdot {\cal D} + \zeta_i, 
\quad |\zeta_i| \leq 16 c n \cdot \eps^{2L-1}, 0 \leq \alpha \leq 2^{-n}.
$$
Choosing $\sigma = \eps^L$ implies the each $\lambda_i'$ w.p. at least $1-2^{-n}$ is sampled according to:
$$
\lambda_i' = \mu(\lambda_i,\eps^{2L}) + \zeta_i, \quad
|\zeta_i| \leq \sigma \cdot 16 c n \eps^{L-1}.
$$
Taking the union bound over all $i\in [n]$ then implies the proof.
\end{proof}

\subsection{Moderately Low-Discrepancy}

\begin{lemma}\label{lem:pseudorandom}
\textbf{Low-discrepancy sequence from almost normal vectors}

\noindent
Let $B>0$, and $v = (v_1,\hdots,v_n)$ be some $(\sigma,\eps)$-normal vector, 
for $\sigma = n^{-B} , \eps \leq n^{-0.9 B}$.
There exists $M \leq n^{1.6 B}$ such that
for any $S = \{i_1,\hdots,i_s \}\subseteq [n]$ , $|S| = s$ the 
distribution on
$s$-dimensional sequence of length $M$:
$$
V_s \equiv
\left\{
\left(
\left\{m \cdot v_{i_1} \right\},\hdots,\left\{m\cdot v_{i_s} \right\}
\right)
\right\}_{m\in [M]}
$$
satisfies
$$
D_M(V_s) \leq 4 \log^s(n) \cdot n^{-0.1 B}. 
$$
\end{lemma}

\begin{proof}

Let $P$ be the minimal prime at least $n^{0.3 B}$, and put $M = P^5$.
By Bertrand's postulate, for sufficiently large $n$ we have that $M = P^5 \leq n^{1.51 B} \leq n^{1.6 B}$. 
For any $z\in [0,1)$ let $z^M$ be the number closest to $z$ in the grid $m/M$, $m\in [M]$.

\paragraph{Removal of non-independent component.}
Since $v$ is $(\sigma,\eps)$-normal then
 $v_i = X_i + Y_i$, where $X_i \sim (\eta_i, \sigma^2), |Y_i| \leq \eps \sigma$, and the $X_i$'s are independent.
Let $V_S^X$ denote the sequence generated by taking only the $X$ component of the seed vector $v$, i.e.:
\begin{equation}\label{eq:ev2}
V_S^X \equiv
\left\{
\left(
\left\{m \cdot X_{i_1} \right\},\hdots,\left\{m\cdot X_{i_s} \right\}
\right)
\right\}_{m\in [M]}
\end{equation}
\begin{fact}\label{fact:xy}
$$
D_M(V_S) 
\leq 
D_M(V_S^X)  + s \cdot n^{- 0.2 B}
$$
\end{fact}

\begin{proof}
Consider the r.v.'s $X_i,Y_i$.
By our assumption 
\begin{equation}
\forall i\in [n] \quad
|Y_i| \leq \sigma \eps = n^{- 1.9 B}.
\end{equation} 
Thus:
\begin{equation}
\forall m\in [M],\  i\in [n] \quad
\left| \left\{m v_i \right\} - \left\{m X_i\right\}\right| \leq m\cdot n^{- 1.9 B} \leq M n^{-1.9 B} \leq  n^{-0.3 B}
\end{equation} 
By Lemma \ref{lem:nied}, we can conclude that the discrepancy of our target sequence $V_S$ 
follows tightly the discrepancy of $V_S^X$:
\begin{equation}\label{eq:closeseq}
D_M(V_S) 
\leq 
D_M(V_S^X)  + s \cdot n^{- 0.3 B}
\end{equation}

\end{proof}

\paragraph{Reducing Gaussian measure to uniform measure}

\noindent
Consider the vector derived by truncating each coordinate of the vector $(X_{i_1},\hdots,X_{i_s})$ to the nearest point on the $M$-grid:
$$
X^M = (X_{i_1}^M,\hdots, X_{i_s}^M)
=
(\round{ M X_{i_1}}/ M,\hdots, \round{ M X_{i_s}}/M).
$$
Consider the discrepancy of the distribution on $s$-dimensional sequences formed by taking
integer multiples of $X^M$. We claim:
\begin{fact}\label{fact:xm}
$$
\P_v \left(D_M(V_S^{X,M}) \leq \log^s(n)\cdot n^{-0.1 B} \right) \geq 1-3 n^{-0.1 B},
$$
\end{fact}

\begin{proof}
In Fact \ref{fact:uninormal} choose as parameter $m=n^{0.2 B+2}$. 
We get that 
w.p. at least  $1-2  n^2/m = 1- 2 n^{-0.2 B}$
each $X_i$ samples a convex mixture of variables $\{w_j\}_{j\in [m]}$ where 
\be\label{eq:Ii}
w_j \sim U(I_j), 
|I_j| = \sigma/m = n^{-1.2 B-2}
\ee
Hence, w.p. at least $1-2 n^{-0.2 B}$ for all $i\in [n]$, the variable $M \cdot \{X_i^M\}$ is a convex mixture
of uniform random variables on intervals 
$I_j\subseteq [M]$, where
\be\label{eq:interval}
|I_j| \geq \frac{\sigma M}{m} \geq n^{1.5 B} \cdot n^{-1.2 B - 2} \geq M^{0.2}.
\ee
We apply Lemma \ref{lem:discrepancy} to the sequence of residuals of integer multiples,
with the seed $X^M$:
\begin{equation}
V_S^{X,M} \equiv \left( \left\{m X_1^M\right\},\hdots,  \left\{m X_s^M\right\}\right)_{m\in [M]}.
\end{equation}
The lemma requires that each variable be distributed as:
$$
M X_i^M\sim U[{\cal I}],
$$ 
where ${\cal I}$ is some interval of $[M]$, for integer $M>1$ satisfying:
$$
|{\cal I}| \geq P, \ \ P \mbox{ prime }, \ \ P \geq N^a,  a>0.
$$
By our choice of parameters $M$ has a prime divisor $P$ equal to $M^{0.2} = P$.
Hence, by Equation \ref{eq:interval} we can satisfy the assumption of the lemma by choosing the parameters
$N,M,a$ for Lemma \ref{lem:discrepancy} as follows:
\be
N=M, M = P, a = 0.2.
\ee
Hence, by Lemma \ref{lem:discrepancy}, and accounting for the Gaussian-to-uniform
approximation error we get:
\begin{equation}\label{eq:discapply}
\P_v \left(D_M(V_S^{X,M}) \leq 2\log^s(n) \cdot n^{-0.1 B} \right) \geq 1-n^{-0.1 B} - 2 n^{-0.2 B} \geq 1 - 3 n^{-0.1 B}.
\end{equation}
\end{proof}

\paragraph{Treating the residual w.r.t. the $M$-grid}
Define: the truncation error 
$$
\forall i\in [s] \quad
r_i := X_i - X_i^M.
$$
In Fact \ref{fact:xy} we analyzed the error $Y_i$ whose magnitude is negligible even w.r.t. $1/M$, and can thus
be disregarded for any element of the sequence $V_S$.
Unlike this,  the residual
error $r_i$ cannot be disregarded because when multiplied by integers uniformly in $[M]$ it assumes magnitude $\Omega(1)$.
Thus, it requires a different treatment.

\begin{corollary}\label{cor:vsx}
$$
\P_v \left( D_M(V_S^{X}) \leq 2\log^s(n)\cdot n^{-0.1 B} \right) \geq 1 - 4 n^{-0.1 B}
$$
\end{corollary}

\begin{proof}
Express the $i$-th element of the sequence using $r_i$:
\begin{equation}
\forall i\in [s] \quad
\left\{X_i \cdot m \right\}
=
\left\{
(X_i^M + r_i)\cdot m
\right\}
=
\{ \{ m X_i^M \} + \{ m r_i \} \}
\end{equation}
Let $E$ denote the event in which $X_i$ is sampled according to $w_j\sim U[{\cal I}_j]$ 
where $w_j$ is at distance at least $1/M$
from either one of the edges of ${\cal I}_j$.
Conditioned on $E$, the random variables $r_i$ and $X_i^M$ are independent for all $i\in [s]$
so conditioned on $E$ we also have
$$
\forall i\in [s] \quad
\{m X_i^M\}, \{m r_i\} \mbox{ are independent. }
$$
$V_S^{X,M}$ is a distribution on sequences formed by sampling the initial
seed $\{X_i^M\}$, and $V_S^X$ is a distribution on sequences formed by
adding to $V_S^{X,M}$ a random sequence formed by sampling the independently seed $\{r_i\}_{i\in [s]}$,
and then generating the length-$M$ sequence 
$$\{(m r_1,\hdots, m r_s)\}_{m\in [M]}.$$
Then since conditioned on $E$ $\{m X_i^M\}, \{m r_i\}$ are independent for all $i$, then by Fact \ref{fact:conv} we have:
$$
D_M(V_S^X|E) = D_M(V_S^{X,M})
$$
and so by Fact \ref{fact:xm} 
\begin{equation}
\P_v \left(D_M(V_S^{X}|E) \leq \log^s(n)\cdot n^{-0.1 B} \right) \geq 1-3 n^{-0.1 B},
\end{equation}
By Equation \ref{eq:interval} the probability of $E$ is at least:
$$
\P_v(E) \geq 
1 - |{\cal I}_j|/(2M) \geq
1- M^{0.2}/(2M) \geq 1 - n^{-B}.
$$
Thus:
\begin{equation}
\P_v \left(D_M(V_S^{X}) \leq \log^s(n)\cdot n^{- 0.1 B} \right) \geq 1-3 n^{-0.1 B} - \P(E) \geq 1 - 4 n^{-0.1 B}.
\end{equation}

\end{proof}

\item
\textbf{Conclusion of proof:}

\noindent
By Corollary \ref{cor:vsx} we have
$$
\P_v \left( D_M(V_S^{X}) \leq 2\log^s(n)\cdot n^{-0.1 B} \right) \geq 1 - 4 n^{-0.1 B}
$$
and by Fact \ref{fact:xy} we have
$$
D_M(V_S) 
\leq 
D_M(V_S^X)  + s \cdot n^{- 0.2 B}
$$
Thus by the union bound:
$$
\P_v 
\left(
D_M(V_S) 
\leq 
2\log^s(n)\cdot n^{-0.1 B} + s \cdot n^{- 0.2 B}
\right) 
\geq 1- 4 n^{-0.1 B}
$$
thus:
$$
\P_v 
\left(
D_M(V_S) 
\leq 
3\log^s(n)\cdot n^{-0.1 B} 
\right) 
\geq 1- 4 n^{-0.1 B}
$$
Hence for all but a measure $4 n^{-0.1 B}$ of sampled vectors $v$, the resulting sequence has discrepancy at most $3 \log^s(n) n^{-0.1 B} $.
Since the discrepancy measures the worst-case additive error for any set this implies that:
$$
D_M(V_S)  \leq 3\log^s(n) n^{-0.1 B} + 4 n^{-0.1 B}
\leq 4 \log^s(n) n^{-0.1 B}
$$
\end{proof}

\subsection{Approximate Pairwise Independence}

\begin{lemma}\label{lem:pairwise}
Let $\bar{\lambda} = (\lambda_1,\hdots, \lambda_n)\in [0,1)^n$ and $M$ a positive integer that satisfy:
$$
\forall i\neq j \quad
D_M \left( \{(m \lambda_i, m \lambda_j)\}_{m\in [M]} \right) \leq \zeta,  \ \ \zeta \leq n^{-4}
$$
Let $4 < \alpha \leq n$.
For each $k\in [n]$ w.p. at least $1/(5 \alpha n )$ over choices of $m\sim U[M]$ the sampled sequence $m$ separates $k$
w.r.t. $B_{in}(\alpha),B_{out}$.
\end{lemma}

\begin{proof}
Let $E_i$ denote the following event:
$$
E_i := x_i\in B_{in} \wedge \forall j\neq i \ \ x_j \notin B_{out}
$$
We want to show that
$$
\forall i\in [n] \ \ 
\P( E_i ) \geq \frac{1}{5 \alpha n}.
$$
Let $s$ denote the number of $x_j$'s {\it not} in $B_{out}$:
$$
s = \left| \{ j \ \  j\neq i \ \ x_j \notin B_{out}\}\right|
$$
Then under this notation we have:
\be\label{eq:ei}
\P(E_i) = \P(s = n-1 \wedge  x_i\in B_{in}).
\ee
Consider the conditional expectation:
$$
\mathbf{E}\left[ s | x_i\in B_{in}  \right]
$$
By linearity of expectation:
\be\label{eq:s1}
\mathbf{E}\left[ s | x_i\in B_{in}  \right]
=
\sum_{j\neq i} \P\left[ x_j\notin B_{out} | x_i\in B_{in}\right].
\ee
Considering each summand separately:
$$
\P(x_j\notin B_{out} | x_i\in B_{in}) = \frac{\P(x_j\notin B_{out} \wedge x_i\in B_{in})}{ \P(x_i\in B_{in})}
$$
Using the pairwise discrepancy assumption, the above is at most:
$$
\frac{ (1 - |B_{out}|) \cdot |B_{in}| + \zeta}{ |B_{in}| - \zeta}
\leq
1 - |B_{out}| + 2 \zeta \alpha n
=
\frac{1}{2n} + 2 \alpha \zeta n \leq \frac{0.51}{n}
$$
and so by Equation \ref{eq:s1}
$$
\mathbf{E}\left[ s | x_i\in B_{in}  \right]  = 
(n-1) \cdot \P\left[ x_j\notin B_{out} | x_i\in B_{in}\right]
\leq 0.51.
$$
The variable $s | x_i\in B_{in}$ accepts only integral values, and by Markov's inequality:
$$
\P( s\geq 1 | x_i\in B_{in}) \leq 0.51
$$
Therefore
$$
\P( s=0 | x_i\in B_{in}) \geq 0.49.
$$
Using again the $1$-dimensional discrepancy we have
$$
\P( x_i\in B_{in}) \geq \frac{1}{\alpha n} - \zeta \geq \frac{1}{2 \alpha n} .
$$
Substituting the last two inequalities into Equation \ref{eq:ei} yields:
$$
\P(E_i) \geq 0.49 \cdot \frac{1}{2 \alpha n} \geq \frac{1}{5 \alpha n}
$$
\end{proof}

\subsection{Proof of Theorem \ref{thm:sep}}

\begin{proof}

By assumption $A$ is $\delta$-separated and
$\zeta \leq \min\{n^{-50},\delta^{13}\}$.
Consider the perturbed matrix 
$$
A' = A + \zeta {\cal E}.
$$
Choose $L =  13$ and $\eps = \zeta^{1/13}$.
Then by Fact \ref{fact:perturb} w.p. at least $1 - n 2^{-n}$ all eigenvalues of $A'$ are $(\sigma,\eps)$-normal 
with parameters 
$$
\sigma  = \zeta \leq n^{-50} , \eps  \leq 16 c n \zeta^{12/13} \leq \zeta^{0.9}= \sigma^{0.9}
$$
where the last inequality follows because $\zeta \leq n^{-50}$.
We assume that this is the case and account for the negligible error at the end.
Set now $B = \log_n(1/\sigma)$.
Then the eigenvalues of $A'$ are $(\sigma,\eps)$-normal for $\sigma = n^{-B}$
an $\eps \leq n^{-0.9 B}$.
Since in addition $4 < \alpha \leq n$ then by Lemma \ref{lem:pseudorandom}
there exists an integer 
$
M \leq n^{1.6 B}
$ 
satisfying:
\be\label{eq:ds1}
\forall S\subseteq [n] \ \ 
D_M( \{m \lambda_{S}\} ) \leq 4 \log^s(n) n^{-5} \leq n^{-4},
\ee
for sufficiently large $n$.
Hence, by Lemma \ref{lem:pairwise} a random $m\sim U[M]$ separates the $k$-th 
eigenvalue of $A + {\cal E}$ w.r.t. $B_{in}(\alpha), B_{out}$ w.p at least $1 / (5\alpha n)$.

\end{proof}

\subsection{Technical approximations}

\begin{fact}\label{fact:uninormal}

\textbf{Approximating a Gaussian by a convex sum of uniform distributions}

\noindent
Let $g = (g_1,\hdots,g_n)$ be standard Gaussian vector $g\in \R^n$.
Then $g$ is equal to a convex combination of two distributions ${\cal D}_u, {\cal D}_v$ as follows: $(1-p){\cal D}_U + p \cdot {\cal D}_V$,
where ${\cal D}_U$ is the $n$-fold distribution of independent variables $z_1,\hdots,z_n$,
where each $z_i$,  $|z_i| \leq \sqrt{n}$, and $z_i$ is  the convex sum $z_i =\sum_{j=1}^m t_j w_j$ of $m\geq 2n^2$ i.i.d. uniformly distributed variables,
with $w_j \sim U[I_j]$, with $|I_j|= 1/m$, and $p \leq 2n^2/m$.
\end{fact}

\begin{proof}
Partition the interval $[-\sqrt{n}/2,\sqrt{n}/2]$ into $m\cdot \sqrt{n}$ equal intervals $I_j$, each of size $1/m$, 
and let $p_j$ denote the point of minimal absolute value in the $j$-th interval.
Set $t_j = \Phi(p_j)$.
Consider the real Gaussian PDF $\Phi(x) = \frac{1}{\sqrt{2\pi}} e^{- x^2 / 2}$.
For any pair of neighboring intervals $p_k,p_{k+1}$, $|p_{k+1}|>|p_k|$, we have  
\be\label{eq:pk1}
\left|\Phi(p_j) - \Phi(p_{j+1})\right| 
\leq 
\Phi(p_j) \cdot \left(1-
\frac{e^{-(p_{j+1}/m)^2/2}}{e^{-(p_j)^2/2}}
\right) 
\leq
1-e^{-p_j/m}
\leq
1 - e^{-\sqrt{n}/m}
\leq
\sqrt{n}/m.
\ee
Then by the above:
$$
\forall j\in [m] \quad
\max\left\{|t_{j}-t_{j-1}|,|t_{j}-t_{j+1}|\right\} \leq \sqrt{n}/m,
$$
Let $z_i: \R\mapsto [0,1]$ denote the following function, which is a sum of uniform distributions
on the intervals $I_j$:
$$
z_i = \sum_{j=1}^{m\sqrt{n}} t_j U[I_j].
$$
Consider the $l_1$-distance between $g_i$ and $z_i$:
$$
\int_{x\in \R} \left|g_i(x) - z_i(x) \right| dx
=
\int_{x\in \R}
\left|
g_i(x)
-
\sum_{j=1}^{m\sqrt{n}} 
\Phi(p_j) U[I_j](x)
\right| dx
$$
By Equation \ref{eq:pk1}, and by bounding the Gaussian tail at $\sqrt{n}$ we establish the upper bound:
$$
\leq
\frac{1}{m}
\sum_{j=1}^{m\sqrt{n}} 
\max\left\{|t_j-t_{j+1}|,|t_j-t_{j-1}|\right\}
+ 
2\cdot e^{-n}
\leq 2 n/m.
$$
On the other hand, by definition of the points $t_j$, we have  
\be
g_i(x)-z_i(x) >0,\ \ \forall x\in \R.
\ee
Hence each $g_i$ may be written as a convex combination 
$$
g_i = (1-p) \cdot \hat{z}_i + p \cdot y_i,
\quad
\hat{z}_i(x) = \frac{z_i(x)}{\int_{\R} z_i(x) dx}
\quad
p \leq 2 n/m.
$$
Since the variables $g_i$ are independent
then the $n$-fold distribution of $g_1,\hdots,g_n$, can be written as a convex combination 
of the of the $n$-th fold distribution of such i.i.d. variables $\hat{z}_1,\hdots, \hat{z}_n$, and another distribution, ${\cal D}_v$ occurring, by the union bound, w.p. at most $n \cdot p \leq 2n^2/m$.
\end{proof}


\section{Parallel Algorithm for ${\rm ASD}$}

The algorithm ${\rm Filter}(A,m,\delta)$ described 
in Section \ref{sec:filter} is given
an integer $m$ that separates the $i$-th eigenvalue, and returns
an approximation for the $i$-th eigenvector.
In this section, we use this algorithm in a black-box fashion and design a
Las-Vegas algorithm for computing the full $\ASD$ of a matrix.
Essentially, it amounts to running sufficiently many copies of the filtering algorithm in parallel
so that all eigenvectors are (coupon) collected with high probability.

\noindent
\begin{mdframed}

\begin{algorithm}\label{alg:main}

\item
Input: $n\times n$ Hermitian matrix $A$, parameter $\delta \leq 1/n$.
\noindent
\begin{enumerate}
\item
Compute parameters:
$$
B = \min\{\delta, B^*(\delta/(3 \sqrt{n}))\},
\delta' = (\min\{\delta,B\})^{13}/ 4
$$
$$
\alpha = \sqrt{\ln(1/\delta')},
M = (\max\{ B^{-12}, n^{-50} \})^{1.6},
{\cal T} = 60 n \alpha \log(n)
$$
$$
B_{out} = [1 - 1/(2n), -1 + 1/(2n)]
\mbox{\quad and \quad}
B_{in} = [-1/(n \alpha ), 1/( n \alpha )],
$$
\item
Perturb:
$A_1 = A + \sqrt{ (\delta')^2 + \delta^2/(9n)} \cdot {\cal E}$, where ${\cal E}$ is GUE.
\item\label{it:parallel}
Run ${\cal T}$ parallel processes of the following procedure
\begin{enumerate}
\item
Sample $m\sim U[M]$
\item
Run Filter $(A_1,m,\delta')$.
\end{enumerate}
\item
For vector $w = w_k$ sampled at process $i\in [{\cal T}]$, compute 
$z = A \cdot w$, $i_0 = \arg\max_{i\in [n]} |w_i|$ and
$
\tilde \lambda_k  = z_{i_0} / w_{i_0}.
$
\item\label{it:dbase1}
Sort the values $\tilde \lambda_i$: assume $\tilde \lambda_1 \leq \hdots \leq \tilde \lambda_{\cal T}$.
Initialize: $\gamma = \tilde \lambda_1$, ${\cal D} = \Phi$.
Iterate over all $i = 1,\hdots, {\cal T}$.
At each step $i$:
if $| \gamma - \tilde \lambda_i | \geq B/4$ then add ${\cal D} \to {\cal D} \cup \{w_i\}$,
and set $\gamma = \tilde \lambda_i$.
\end{enumerate}

\end{algorithm}

\end{mdframed}

%

\noindent
\\
We now state our main theorem:
\begin{mdframed}
\begin{theorem}
For any Hermitian matrix $0 \preceq A \preceq 0.9 I$, and $\delta \leq n^{-10}$ there exists
an $\RNC^{(2)}$ algorithm computing ${\rm ASD}(A,\delta)$,
in boolean complexity $\tilde{O}(n^{\omega+1})$.
The algorithm is log-stable.
\end{theorem}
\end{mdframed}

\begin{proof}

\item
\textbf{Correctness:}

\noindent
Let ${\cal E}_1, {\cal E}_2$ be independent GUE matrices, and set
$$
A_0 = A + (\delta/(3 \sqrt{n})) \cdot {\cal E}_1.
$$
By Lemma \ref{lem:NTV} and definition \ref{def:bstar} the parameter $B \leq B^*$ above satisfies:
$$
\P \left(A_0 \mbox{ is } B \mbox{ separated}\right) \geq 
\P \left(A_0 \mbox{ is } B^* \mbox{ separated}\right) \geq
0.99.
$$
Assume that this is the case, and account for the error at the end.
Consider now a small-scale perturbation of $A_0$:  
\be\label{eq:a0}
A_1 = A_0 + \delta' \cdot {\cal E}_2.
\ee
Since ${\cal E}_1,{\cal E}_2$ are independent the matrix $A_1$ is a perturbation of $A$ as follows:
\be\label{eq:perturb1}
A_1 = A + \sqrt{ (\delta')^2 + (\delta/(3 \sqrt{n}))^2 } \cdot {\cal E},
\ee
for GUE matrix ${\cal E}$.
Hence the matrix $A_1$ used by the algorithm starting at step \ref{it:parallel} is a $\delta'$-perturbation
of a $B$-separated matrix $A_0$.

We now invoke Theorem \ref{thm:sep} w.r.t. $A_0$.
$A_0$ is separated with parameter $B$, and perturbed by $GUE$ with standard deviation
$\zeta = \delta' \leq B^{13}$.
We also have $\delta \leq n^{-10}$ which means in particular that
the parameter $\zeta$ used in the statement of  Theorem \ref{thm:sep} satisfies $\zeta \leq n^{-50}$.
Finally we have $2 < \alpha = O(\sqrt{\ln(n)}) = o(n)$.
Therefore:
$$
\forall k\in [n] \quad
\P_{{\cal E}_2,m\sim U[M]}
\left(
\mbox{$m$ separates $k$
in $A_1 = A_0 + \delta' {\cal E}_2$ w.r.t. $B_{in}(\alpha),B_{out}$}
\right)
\geq
1 / (5\alpha n)
$$
Let $\{v_i\}_{i=1}^n$ denote an orthonormal basis for $A_1$, and let $\lambda_i$ denote the eigenvalues of $A_1$.

Conditioned on sampling $A_1,m$ such that $m$ separates $k$ in $A_1$, 
we invoke Theorem \ref{thm:filter} for $w = {\rm Filter}(A,m,\delta')$.
By our choice of parameters we have that $\delta' \leq n^{-10}$, and $\alpha = \sqrt{\ln(1/\delta')}$.
Hence, the conditions of the theorem are met and so the output vector $w$ satisfies:
$$
\P_{w_0}
\left(
\left\|
w - v_k
\right\|
\leq \delta'
\right) \geq 1 - 3 n^{-3}.
$$
Thus, by the union bound we have
$$
\forall k\in [n] \quad
\P_{{\cal E}_2,m,w_0}
\left(
\left\|
w - v_k
\right\|
\leq \delta'
\right)
\geq
\frac{1}{5  n \alpha} - 3 n^{-3}   \geq \frac{1}{6  n \alpha}
$$
Therefore, by the coupon collector's bound the probability that ${\cal T} = 60 n \alpha \log(n)$
parallel copies sample an approximation for {\it each} $k\in [n]$ is at least
$$
1 - 6/60  - 0.01 = 0.89
$$

Assume that this is the case.
Let $\{w_i\}_{i=1}^{\cal T}$ denote the output set of vectors from all ${\cal T}$ parallel copies.
We now want to show that each output vector is a valid approximate eigen-vector.
Each output vector satisfies by the stopping criterion:
$$
\exists c \ \ 
\left\| \tilde \lambda_i \cdot w_i - A_1 \cdot w_i \right\| \leq 3\delta' \sqrt{n} \leq B/200,
$$
where the last inequality is for sufficiently large $n$,
and $\tilde \lambda_i$ is the algorithm's approximation of the $i$-th eigenvalue.
By Fact \ref{fact:Stewart} since $A_0$ is $B$-separated then 
\be\label{eq:sepa1}
\P( A_1 \mbox{ is } B/2 \mbox{ separated}) \geq 1 - 2^{-n}.
\ee
Assuming this is the case, 
$\delta' \sqrt{n}$ is negligibly smaller than the separation of $A_1$ - so
by Fact \ref{fact:alpha} the value $\tilde \lambda_i$ above must satisfy:
$$
\exists k\in [n] \ \ 
\left| \tilde \lambda_i - \lambda_k \right|
\leq B/10.
$$
and there exists some eigenvector $v_k$ of $\lambda_k$ such that
$$
\| w_i - v_k \| \leq B/30.
$$
This implies by the triangle inequality that for any $j\neq i, l\neq k$ for which
$
\left| \tilde \lambda_j - \lambda_l \right| \leq B/10
$
we have
$$
\left| \tilde \lambda_i - \tilde \lambda_j \right|
> B/2 - 2 B/10 > B/4.
$$
hence $\tilde \lambda_j$ and $\tilde \lambda_j$ are classified to different eigenvalue bins at step \ref{it:dbase1}.
We conclude that the set of values $\{\tilde \lambda_i\}_{i=1}^{\cal T}$ satisfies that
\begin{itemize}
\item
For every $k\in [n]$ there exists $j\in {\cal T}$ such that:
$
\left|  \lambda_k - \tilde \lambda_j \right| \leq B/10
$
and there exists a sampled vector $w_j$, and a unit eigenvector $v_k$ of $\lambda_k$
such that $\| w_j - v_k \| \leq B/30$.
\item
For every $j\in {\cal T}$ there exists a unique $k\in [n]$ such that
$
\left|  \lambda_k - \tilde \lambda_j \right| \leq B/10
$
\end{itemize}
Therefore, at the end of step \ref{it:dbase1}
we have $|{\cal D}| = n$ and
$$
\forall k\in [n] \quad \exists w\in {\cal D} \ \ 
\left\|
w - v_k
\right\| \leq B/10.
$$
Hence, with probability at least $0.89 - 2^{-n}\geq 0.85$ the algorithm returns a database ${\cal D} = \{w_1,\hdots,w_n\}$,
such that there exists an orthonormal basis $\{v_1,\hdots,v_n\}$ of $A_1$ for which
$$
\forall k\in [n] \ \ 
\| v_k - w_k \| \leq B/10 \leq \delta/10.
$$
In that case ${\cal D}$ is an $\ASD(A_1,\delta/10)$.
On the other and, by Equation \ref{eq:perturb1} 
we have that w.p. $1 - 2^{-n}$ $A_1$ satisfies:
$$
\|A - A_1\|^2 \leq (\delta/3)^2.
$$
Therefore w.p. at least $0.84$ the database ${\cal D}$ is also $\ASD(A,\sqrt{(\delta/10)^2 +   (\delta/3)^2})$ which is in particular
an $\ASD(A, \delta)$.

\item
\textbf{Run time:}

\noindent
By Theorem \ref{thm:filter} the parallel time of each copy of ${\rm Filter}(A,m,\delta')$ is at most 
$O(\log^2(n))$, and arithmetic complexity $\tilde{O}(n^{\omega})$.
Hence, this is the depth complexity of all ${\cal T}$ parallel copies of the circuit for ${\rm Filter}(A,m)$.
The sorting step \ref{it:dbase1} can be implemented in depth $O(\log(n))$ by Fact \ref{fact:sort}, using $O(n)$ 
processors.
Hence the complete algorithm runs in parallel time $O(\log^2(n))$ in arithmetic complexity $\tilde{O}(n^{\omega+1})$.

\noindent
Finally, we show that the entire algorithm can be implemented in $\log(n)$-precision:
\begin{proposition}
The algorithm \ref{alg:main} is log-stable.
In particular it runs in boolean complexity $\tilde{O}(n^{\omega+1})$.
\end{proposition}
\begin{proof}
Each parallel process is log-stable by Claim \ref{cl:stable1}.
Hence, to show that the entire algorithm is log-stable, it is sufficient to show that for a Gaussian perturbation
truncated to $t = O(\log(n))$ bits of precision, the statement of Theorem \ref{thm:sep} still applies - namely
a uniformly random $m\sim U[M]$ separates each eigenvalue, with high probability, for the parameter $M$ set in the algorithm.
To do that, write 
$$
{\cal E}_t = {\cal E} + D,
$$ 
where ${\cal E}$ is the non-truncated GUE matrix,
and $D$ is some matrix sampled from a distribution such that $|D_{i,j}| \leq 2^{-t}$ for all $i,j$. 
Following Equation \ref{eq:a0} define 
$$
A_1^{(t)} \equiv A_0 + {\cal E} + D = A_1 + D,
$$
which is then used by the parallel process of ${\rm Filter}(A_1^{(t)},m,\delta')$.
By Equation \ref{eq:sepa1} w.p. at least $1 - 2^{-n}$ the matrix $A_1$ is separated with parameter at least
$B/2$.
Hence, for sufficiently large $t = O(\log(n))$, such that $2^t > B$ and using again Fact \ref{fact:Stewart} we have for each $i\in [n]$:
$$
\left| \lambda_i(A_1) - \lambda_i(A_1^{(t)})\right| \leq 2\sqrt{n} \cdot 2^{-t}.
$$
Hence by Lemma \ref{lem:nied} if 
$$
( \{ m \lambda_{i_1}(A_1)\},\hdots, \{ m \lambda_{i_2}(A_1)\})_{m\in [M]},
$$ 
is a $2$-dimensional
sequence with discrepancy $D_M$,  then the discrepancy of the noisy sequence:
$( \{ m \lambda_{i_1}(A_1^{(t)})\},\hdots, \{ m \lambda_{i_2}(A_1^{(t)})\})_{m\in [M]}$
is at most $D_M + 2 \cdot M \cdot \sqrt{n} \cdot 2^{-t}$.
This implies that the pairwise discrepancy error increases by an additive error at most $2M \sqrt{n} 2^{-t}$
at Equation \ref{eq:ds1} in the proof of theorem \ref{thm:sep}.
Hence, there exists a choice of $t = O(\log(n))$ so that the discrepancy of the above sequence is at most $n^{-4}$,
thereby satisfying Equation \ref{eq:ds1}.
Hence, the statement of Theorem \ref{thm:sep} holds for ${\cal E} = {\cal E}_t$, for $t = O(\log(n))$.

This implies that the algorithm can be computed in $\tilde{O}(n^{\omega+1})$ arithmetic operation,
where each operation is a binary operation between two registers of size $O(\log(n))$.
Hence, the boolean complexity of the algorithm is $\tilde{O}(n^{\omega+1}) \cdot \log(n) = \tilde{O}(n^{\omega+1})$.
\end{proof}

\end{proof}

\subsection{Supporting Claims}

\begin{fact}\label{fact:alpha}
Let $A$ be a $\delta$-separated matrix, $w$ a unit vector satisfying
$$
\| A w - c w\| \leq B, \ \ B\leq \delta^2/100,
$$
for some $|c| \leq 1$.
Then there exists $\lambda\in {\cal L}(A)$ such that
$$
| c - \lambda | \leq \delta/10,
$$
and there exists some unit eigenvector $v$ of $\lambda$ such that
$$
\| v - w \| \leq \delta/30.
$$
\end{fact}

\begin{proof}

Without loss of generality, we assume $w = \alpha_1 v_1 + \alpha_2 v_2$ where $v_1,v_2$ are two distinct elements of some
orthonormal basis of $A$, corresponding to eigenvalues $\lambda_1,\lambda_2$.
Then
$$
\| A w - c w \|^2
=
\| v_1 \alpha_1 (c - \lambda_1) + v_2 \alpha_2 (c - \lambda_2) \|^2 \leq B^2
$$
By the fact that $v_1,v_2$ are orthogonal:
\be\label{eq:alpha}
|\alpha_1|^2 \cdot |c - \lambda_1|^2 + |\alpha_2|^2 \cdot |c - \lambda_2|^2 \leq B^2.
\ee
Since $\|w\|=1$ then $|\alpha_1|^2 + |\alpha_2|^2 = 1$, so for at least one $i\in \{1,2\}$ we have
$$
|c - \lambda_i|^2 \leq 2 B^2.
$$
Suppose w.l.o.g. that $i=1$. 
Since $A$ is $\delta$ separated then $|\lambda_1 - \lambda_2| \geq \delta$.
Then together with the triangle inequality we have
$$
|\lambda_2 - c| \geq |\lambda_2 - \lambda_1| - |c - \lambda_1| \geq \delta - \sqrt{2} B \geq 0.9 \delta.
$$
Hence by Equation \ref{eq:alpha} above:
$$
|\alpha_2|^2 \leq B^2 / (0.9 \delta)^2 \leq \delta^2 / 1000
$$
This implies that for some unit eigenvector $v$ of $\lambda_1$ we have
$$
\| w -  v \| \leq \delta/30,
$$
Using this inequality back in the assumption, and since $\|A\|\leq 1, |c| \leq 1$ then
$$
\| c w - A w \| = \| c(v + {\cal E}) - A \cdot (v + {\cal E}) \| \leq B, \quad \|{\cal E}\| \leq \delta/30
$$
and since $A v = \lambda_1 v$ then
$$
\| c v - \lambda_1 v \| \leq B + 2 \delta/30,
$$
and since $\|v\|=1$ we get:
$$
| c - \lambda_1 | \leq B + 2 \delta/30 \leq \delta/10.
$$
\end{proof}


\section{Numerical Experiments}\label{sec:numerics}

In the work above we have outlined a new theoretical approach for solving the eigen-problem of Hermitian matrices.
Given the immense practical importance of this problem, we now provide in addition some numerical evidence
about the performance of the algorithm.
We run the main algorithm for $ASD$ of a random matrix $A$ for $n = 20$, where the number of bits of precision
used is standard Matlab single precision of $32$ bits, and require output precision of $\delta = 10^{-4}$.
We then run $50 \cdot n \cdot \log n \sim 2500$ copies of the filter procedure $Filter(A,m,\delta)$.
Then, we iterate over a sequence of values for the integer $M$, i.e. the number such that such that by 
Theorem \ref{thm:sep} a uniformly random integer $m\sim U[M]$ separates any eigenvalue of a
perturbed version of $A$ w.h.p.

\begin{figure}
\center{
 \epsfxsize=4in
 \epsfbox{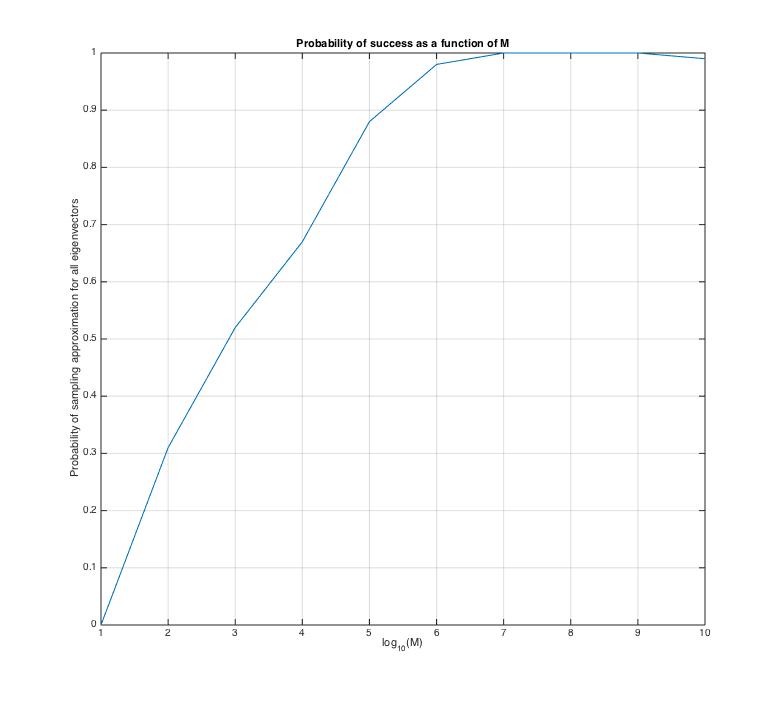}
 \caption{\footnotesize{Algorithm Behavior}}}
\end{figure}

One can see that while the analytic behavior requires $M$ to be at least $n^{75}$, here we see that even $M = n^5$ already
suffices for the algorithm to return an approximate eigenvector for each eigenvalue w.p. very close to $1$.
We conjecture, hence, that implementing the algorithm will result in far better practical run time than the analytical bounds we provide here.

\section*{Acknowledgements}
\noindent
The authors thank Naomi Kirshner and Robin Kothari for very helpful comments, and Yosi Atia for pointing to us an error in an earlier version of this paper.
We also thank anonymous reviewers for their helpful comments and suggestions.
This research project was supported in part by the Israeli Centers of Research Excellence (I-CORE) program (Center  No. 4/11), by the Israeli Science Foundation (ISF) research grant 1446/09, by an EU FP7 ERC grant (no.280157), and by the EU FP7-ICT project QALGO (FET-Proactive Scheme).
LE is thankful to the Templeton Foundation for their support of this work.


\begin{thebibliography}{LMM99}

\bibitem[Akl90]{AKL}
Selim~G. Akl.
\newblock {\em Parallel Sorting Algorithms}.
\newblock Academic Press, Inc., Orlando, FL, USA, 1990.

\bibitem[BP92]{Bini92}
D.~Bini and V.~Pan.
\newblock Practical improvement of the divide-and-conquer eigenvalue
  algorithms.
\newblock {\em Computing}, 48(1):109--123, 1992.

\bibitem[DDH07]{DDH07}
James Demmel, Ioana Dumitriu, and Olga Holtz.
\newblock Fast linear algebra is stable.
\newblock {\em Numerische Mathematik}, 108(1):59--91, 2007.

\bibitem[GVL96]{Golub}
Gene~H. Golub and Charles~F. Van~Loan.
\newblock {\em Matrix Computations (3rd Ed.)}.
\newblock Johns Hopkins University Press, Baltimore, MD, USA, 1996.

\bibitem[HHL09]{HHL}
Aram~W. Harrow, Avinatan Hassidim, and Seth Lloyd.
\newblock Quantum algorithm for linear systems of equations.
\newblock {\em Phys. Rev. Lett.}, 103:150502, Oct 2009.

\bibitem[Koz92]{Kozen}
Dexter~C. Kozen.
\newblock {\em The Design and Analysis of Algorithms}.
\newblock Springer-Verlag New York, Inc., New York, NY, USA, 1992.

\bibitem[LMM99]{LMM}
Mauro Leoncini, Giovanni Manzini, and Luciano Margara.
\newblock Parallel complexity of numerically accurate linear system solvers.
\newblock {\em SIAM Journal on Computing}, 28(6):2030--2058, 1999.

\bibitem[Nie92]{Nied}
H.~Niederreiter.
\newblock {\em Random Number Generation and Quasi-Monte Carlo Methods}.
\newblock Society for Industrial and Applied Mathematics, 1992.

\bibitem[NTV16]{NTV}
Hoi Nguyen, Terence Tao, and Van Vu.
\newblock Random matrices: tail bounds for gaps between eigenvalues.
\newblock {\em Probability Theory and Related Fields}, pages 1--40, 2016.

\bibitem[Rei05]{Reif}
John~H. Reif.
\newblock Efficient parallel factorization and solution of structured and
  unstructured linear systems.
\newblock {\em Journal of Computer and System Sciences}, 71(1):86 -- 143, 2005.

\bibitem[SS90]{Stewart}
Gilbert~W. Stewart and Jiguang Sun.
\newblock {\em Matrix perturbation theory}.
\newblock Computer science and scientific computing. Academic Press, Boston,
  1990.

\bibitem[TB97]{Trefethen}
Lloyd~N. Trefethen and David Bau.
\newblock {\em Numerical linear algebra}.
\newblock Society for Industrial and Applied Mathematics, Philadelphia, 1997.

\bibitem[TS13]{TaShma}
Amnon Ta-Shma.
\newblock Inverting well conditioned matrices in quantum logspace.
\newblock In {\em Proceedings of the Forty-fifth Annual ACM Symposium on Theory
  of Computing}, STOC '13, pages 881--890, New York, NY, USA, 2013. ACM.

\bibitem[Wil12]{Williams}
Virginia~Vassilevska Williams.
\newblock Multiplying matrices faster than coppersmith-winograd.
\newblock In {\em Proceedings of the Forty-fourth Annual ACM Symposium on
  Theory of Computing}, STOC '12, pages 887--898, New York, NY, USA, 2012. ACM.

\end{thebibliography}
\end{document}